\documentclass[manuscript,screen,nonacm,11pt]{acmart}

\setcopyright{none}
\acmDOI{}

\usepackage{booktabs}
\usepackage{pgfplots}
\pgfplotsset{compat=1.15}

\newtheorem{theorem}{Theorem}
\newtheorem{lemma}{Lemma}
\newtheorem{proposition}{Proposition}
\newtheorem{corollary}{Corollary}
\newtheorem{definition}{Definition}
\newtheorem{axiom}{Axiom}
\newtheorem{conjecture}{Conjecture}
\newtheorem{remark}{Remark}
\newtheorem{prediction}{Prediction}
\newtheorem{example}{Example}

\newcommand{\csm}{\textsc{csm}}
\newcommand{\act}{\mathsf{A}}
\newcommand{\traf}{\mathsf{F}}
\newcommand{\cw}{\mathrm{vcw}}
\newcommand{\OO}{O}
\newcommand{\pJ}{\,\mathrm{pJ}}

\begin{document}

\title{The Price of Remembering: A Calibrated Energy Law for Computation}

\author{Mohamed Amine Bergach}
\email{mbergach@gmail.com}
\affiliation{%
  \institution{Illumina}
  \city{San Diego}
  \state{California}
  \country{USA}
}
\renewcommand{\shortauthors}{}

\begin{abstract}
Where does a computer's energy go? Mostly into keeping, not into
computing. A bit held in fast storage draws power for every second it
stays there, and it costs energy again each time it moves between
storage levels. We call the first cost \emph{rent} and the second
\emph{fare}, and we state one law: the energy of a computation is at
least its operations, plus rent on every live bit for as long as it
lives, plus fare on every bit moved. The model under the law prices
control as well as data. There is no free clock, and any unpriced
register would make the theorems false. One lemma does most of the
work: every use of a value is paid for by rent, by fare, or by
computing the value again. Three things follow. Exact attention
brings every past token back for every new one, so its energy grows
with the square of the context length, while a recurrent model with a
fixed state grows linearly. The square is a theorem for machines that
never re-read past tokens. Under a stated serving hypothesis it is
the fare on every past token, which passes the model's own arithmetic
near ten thousand tokens, the point where long-context serving
becomes bandwidth-bound today. Known bounds on memory over time become
joule floors: on any sequential machine with volatile working
storage, sorting $n$ items pays rent proportional to $n^2/\log n$
bit-steps on most inputs, and the bound for scrypt makes every
password guess cost joules that no amount of parallel hardware
reduces. The law is calibrated. On a synthesized 45\,nm processor
whose storage physically moves, the rent constant measured at gate
level is a clock term, $0.82$\,pJ per occupied slot per cycle, plus
$3.4$\,fJ per net transition, within $4\%$ on every rotating run on
bare or operand-isolated hardware. Four fifths of that rent is the
clock: keeping is mostly knowing when. On the same instrument,
staging a reduction across rings, the idle rings parked, beat one
flattening of it onto a single ring $83\times$ in rotation action.
We say what would prove the law wrong.
\end{abstract}

\maketitle

\section{Introduction}
\label{sec:intro}

A 64-bit addition at 45\,nm costs about a picojoule. Fetching its two
operands from off-chip memory costs a thousand to ten thousand times
more~\cite{Horowitz2014,KecklerDally11}. In a general-purpose core,
over $90\%$ of the energy of an instruction goes to overhead (fetching,
decoding, renaming, tagging, forwarding), not to the instruction
itself~\cite{Hameed10}. Engineers know this and price it in tables. We
state it as a law with lower bounds. (Throughout, ``measured'' means
extracted from gate-level simulation of a synthesized, pre-layout
netlist; we fabricated no silicon. Section~\ref{sec:instrument} states
the caveats.)

The law has three terms beyond the arithmetic. A computation pays
\emph{rent} on every bit it keeps: fast storage burns power for every
second a bit sits in it, whether or not the bit is touched, as leakage
in SRAM, as refresh in DRAM, and as physical rotation in the machine we
calibrate on. It pays \emph{fare} on every bit it moves between storage
levels, as a cache fill or a memory bus does. And it pays, far below
either, Landauer's toll on every bit it erases~\cite{Landauer61}. A
plain example: a phone number you will need again in a minute can be
kept in mind (rent), read again from the phone (fare), or worked out
again from wherever you first got it (recompute). Everything below is a
careful version of that choice.

The model beneath the law has one rule that sets it apart, and we adopt
it as a slogan: \emph{there is no free clock}. A machine's next step is
a function of the contents of its computed cells alone
(Definition~\ref{def:substrate}). Its schedule, its place in its own
program, and its sense of time are bits like any others, and they pay
rent while they are held. The rule is forced, not fussy. Give the model
any unpriced register, elapsed time included, and the theorems below
become false. The input is already held free in the environment, so
hiding it would buy nothing; what an unpriced clock hides is the
working set, the counter a scheduler keeps or the ranks a sorter has
learned so far, which a machine could keep in how long it waits
between steps and read back through its control, paying no rent on
exactly what Corollary~\ref{cor:sortjoules} charges. The same test forces the rule's second half, that the
state describes itself. A machine free to record a bit in \emph{which}
of two cells it occupies, rather than in what a cell holds, keeps its
working set in placement, one live bit among $2^m$ cells carrying $m$
bits rent-free, and again sorts for almost nothing. So in
Definition~\ref{def:substrate} a computed state is determined by its
contents read in a fixed order of the cells, and which cells are
occupied is information those contents already carry. The rule also
has a measured face: four fifths of the rent constant of
Section~\ref{sec:instrument} is the clock delivered to an occupied
slot's flops, the cost of live state knowing when it is, and the floor
an empty slot pays is the rent of the cell that records its emptiness,
the cost of the ring knowing where its packets are. The two are an
echo rather than an identity, since the clock term prices every live
flop, data included; but theory and netlist agree on the sentence that
keeping is mostly knowing when. Section~\ref{sec:choice} names what
our extraction leaves unpriced, the clock tree itself, as the open
capability.

The striking fact is that the rent functional, live bits summed over
time, which we call the \emph{action} of the computation, is already
the billing unit of four communities that do not cite one another.
Memory-hard cryptography prices ASIC adversaries by \emph{cumulative
memory complexity}, exactly $\sum_t \text{space}(t)$, and has built a
lower-bound machinery for it~\cite{AlwenSerbinenko15,AlwenBlocki16,
AlwenBlockiPietrzak18,RenDevadas17,BlockiRenZhou18}. Complexity theory
recently imported the measure and proved that sorting requires
cumulative memory $\tilde\Omega(n^2)$, motivated by cloud platforms
that bill in gigabyte-seconds~\cite{BeameKornerup23}. Computer
architecture tabulates the fare term (picojoules per bit per
millimeter, per cache level, per DRAM
interface)~\cite{KecklerDally11,Horowitz2014} and prices the rent term
in every DRAM refresh-power specification~\cite{LiuJVM12}. The theory
of energy-efficient algorithms prices the erasure
term~\cite{DemaineLMT16}. Valiant has asked for a fine-grained
\emph{energy-centric} counterpart of the Church--Turing
thesis~\cite{Valiant24}. We propose one: the four bills are one
functional, the functional has a physical constant, and lower bounds
on it are energy lower bounds that no algorithmic or architectural
cleverness evades, only trades.

\subsection{Results}
\begin{itemize}
\item \textbf{The law (Section~\ref{sec:law}).} Over leveled volatile
  substrates (levels with capacities, hold powers $p_\ell$, and boundary
  fares $c_\ell$, under a volatility axiom grounded in device physics)
  energy decomposes as operations plus rent $\sum_\ell p_\ell \act_\ell$
  plus fare $\sum_\ell c_\ell \traf_\ell$; Landauer's erasure toll sits
  far below both (Remark~\ref{rem:erasure}). Our workhorse is a
  \emph{serving lemma}: every use of a value is served by residence since
  its last use (rent), by a fetch (fare), or by recomputation
  (operations), with disjoint charges. The model bills control as well
  as data: the next step is a function of priced state alone, read as
  contents, and with any unpriced register the bridges below fail
  (Definition~\ref{def:substrate}).
\item \textbf{Bridges (Section~\ref{sec:bridges}).} Existing lower
  bounds become joule floors. Through a simulation lemma, the
  cumulative-memory bounds of \cite{BeameKornerup23} give: any
  one-station machine sorting $n$ items on a substrate whose working
  levels are volatile pays rent on its working set $\geq
  p_{\min}\tau_0\,\Omega(n^2/\log n)$ on most inputs, with $\tau_0$
  the fast level's access time, no synchrony hypothesis, and
  the input held rent-free and re-readable. The floor is on the working
  set, not on holding the input, which is what makes it a theorem, and
  it sits at the picojoule scale where machines live, not at $kT\ln 2$.
  The cumulative-memory bound for scrypt~\cite{ACPRT17} transfers with
  no loss to parallelism: every machine evaluating it through oracle
  stations, at any width and whatever its holds, pays rent $\geq
  p_{\min}\tau_0\,\Omega(n^2 w)$ per guess, the
  ASIC adversary's joule floor (Corollary~\ref{cor:mhf}). Hong--Kung I/O
  bounds~\cite{HongKung81,IronyToledoTiskin04} give the fare floor
  $\Omega(n^3/\sqrt S)$ for matrix multiplication, which we port to
  moving-storage hierarchies.
\item \textbf{The meeting bound (Section~\ref{sec:meeting}).} Exact
  causal attention over $T$ steps pays, unconditionally in the streaming
  setting, rent $\Omega(p_{\min}\,\Delta t\, d\,T^2)$ on $d$ bits per
  token, through a known space lower bound for one attention function:
  millijoules at today's constants by extrapolation, the bound being
  asymptotic in $d$. With a per-token key--value footprint of $\kappa$
  bits and a route hypothesis it pays $E \geq
  \kappa\,\frac{T(T-1)}{2}\,\gamma$, with $\gamma$ the least of the
  platform's fare, its rent per step, and its recomputation floor: tens
  of kilojoules at today's constants, the square with its constant
  filled in; a capacity form needs no timing. A recurrent state of
  $\sigma$ bits pays a memory term linear in $T$, so the separation is
  $\Theta(\kappa T/\sigma)$. With today's constants the attention floor
  overtakes the model's arithmetic near $T \approx 1$--$3 \times 10^4$,
  where long-context serving becomes bandwidth-bound in practice.
  FlashAttention, paged KV, and gradient checkpointing are placements on
  the law's frontier, not escapes from it (Figure~\ref{fig:frontier}).
\item \textbf{Structure: why hierarchies exist
  (Section~\ref{sec:structure}).} In the cyclic-storage model that
  abstracts our calibrating machine, a staged hierarchy beats a
  flattened single ring in rotation action by $\Omega(N/(s\,b))$ on
  $N$-leaf reductions, for $s$ stations and a stage size $b$ fixed by
  the window ($b = 4$ at $W = 8$), linear in $N$. Priced, the
  hierarchy's bill is parked rent, and on our instrument's constants
  that rent sits above the single ring's floor at every $N$: the
  theorem separates rotation, not energy, and we say where its
  balance point lies. The instrument's eight-seed instance, staged
  against one relay-ladder flattening of the same program, realizes
  $83\times$ in rotation action. Necessity, by
  contrast, is asymptotic: bounded reach forces value-cutwidth
  $\OO(W)$, so reductions wider than $2^{8W}$ leaves cannot seat on
  one ring at all, a threshold beyond every practical size; optimal
  seating is NP-hard to approximate within any constant; and systolic
  arrays embed exactly (Appendix~\ref{app:proofs}). Economy, not
  feasibility, is why hierarchies exist.
\item \textbf{The dependence tax (Section~\ref{sec:boundary}).} On a
  cyclic store of period $P$, an $L$-step dependent pointer chase costs
  expected time $\geq (L-1)\frac{P}{2}(1-\frac1n)$ whatever the
  placement, and a one-comparator mechanism in our instrument meets the
  bound within one cycle. This prices the boundary of static
  scheduling.
\item \textbf{The instrument (Section~\ref{sec:instrument}).} A
  synthesized 45\,nm ring processor whose storage \emph{moves} makes
  rent and fare one constant. Gate-level extraction from three real
  programs and two rotation runs by the difference method gives $E =
  \text{ops} + \text{fare} + \text{idle} + e\cdot\act$ with $e = 0.82\pJ
  + 3.4\,\mathrm{fJ} \times$ (net transitions per live slot-cycle), a
  clock term that does not depend on the data plus an activity term
  that does, fitting every static run on bare or operand-isolated
  hardware within $4\%$, from $0.90$ to $1.8\pJ$. The law is an
  equality somewhere real, and its constants are physical, not slack.
\item \textbf{Rent on choice (Section~\ref{sec:choice}).} The same
  measurements price \emph{capability}: moving a ring's freedom to
  stand still from a per-bit hold mux to a per-slot clock gate made
  it $28\%$ smaller and $91\times$ cheaper at idle in
  switching energy ($39\times$ with leakage counted), and isolating
  stations from data they may observe cut their churn from $27.0$ to
  $4.1\pJ$/cycle. A provisioned freedom bills every cycle, used or
  not; we state the principle and leave its formalization open.
\item \textbf{Testing the law (Section~\ref{sec:falsify}).} We give
  the one test that can refute the law's form (the two-term constant
  must predict every rotating run; it does, within $4\%$ on five, and
  the run it misses by $20\%$ is the un-isolated hardware, whose
  heavier nets isolation removes), two measurable predictions on its
  constants, and the counterexample surface (fast nonvolatile memory;
  reversible logic); either would reshape the law's terms, and the law
  says how.
\end{itemize}

\subsection{What kind of claim this is}
Like the Church--Turing thesis, the law is axioms plus theorems plus
evidence: axioms about physical substrates (volatility, embedding),
theorems inside the model (bridges, meeting bounds, obstructions,
separations), and measurements that pin the constants. We claim no
progress on classical open problems. We claim that one physically
calibrated functional organizes results now scattered across
cryptography, complexity, architecture, and machine-learning systems,
and that it yields new theorems where they meet. We call \eqref{eq:law}
a law in the sense of Section~\ref{sec:law}: a definition, an axiom
that its fast-level constants are positive, and a measured band. It can
be refuted through Section~\ref{sec:falsify}, and it is not a law in
the sense of a new equation of physics.

\emph{What is new, and where it sits.} The model is
Definition~\ref{def:substrate}. Its four rules are each forced, and we
show how the theorems fail without them: control is a function of
priced state alone; the state describes itself (which cells are
occupied is never information beside their contents); an unpriced
environment tape holds the input; and raw words and output words are
units, so an input bit enters the computed string only through a fire
that read its word and an output leaves as the item the cited bounds
count
(Section~\ref{sec:intro}; Remark~\ref{rem:input}; and, for the fourth,
without it one fire could gather bits of $2w$ input words, or a
partial overwrite could relabel a word's unread bits as computed, and
the simulation of Lemma~\ref{lem:sim} would need $w$ times its levels
or would not exist). The bridges rest on two lemmas the cited bounds do
not supply, a per-input form of Beame--Kornerup's theorem
(Proposition~\ref{prop:perinput}) and an oracle projection into the
parallel random-oracle model (Lemma~\ref{lem:prom}); the meeting bound
has an unconditional streaming form (Theorem~\ref{thm:uncond}); the
serving lemma's accounting is machine-checked, proved in TLAPS in its
atomic form (Appendix~\ref{app:tla}); and the law is calibrated on a
synthesized netlist whose measured constant resolves into clock plus
activity within $4\%$ (Section~\ref{sec:instrument}), a loop from cost
model to gate-level silicon that we have not seen closed elsewhere. The
structural theorems (Section~\ref{sec:structure}) and the instrument
sit at the edge of a theory venue's usual scope; they are in the paper
because the axioms are empirical claims and the instrument is their
evidence. Related work is collected in Section~\ref{sec:related}.

\section{The Law}
\label{sec:law}

\noindent\textbf{The model in plain words.} A machine has a few tiers
of storage, from fast and expensive to slow and cheap, and an unpriced
tape outside them that holds the input and receives the output. Every
bit held in a tier pays rent for as long as it stays. Every bit that
crosses from one tier to the next pays a fare. Operations run at fixed
places called stations and read what sits at their own tier. Time
passes only in holds; between holds the machine may move bits and fire
operations, and those take no time. Two rules keep the accounting
honest: the machine decides its next step from the contents of its
computed cells and from nothing else (no hidden clock, no hidden
counter), and those contents, read in a fixed order, say which cells
are occupied (no hiding a bit in \emph{where} a bit sits). The definition below makes each of these precise, and
Example~\ref{ex:add}, after the action law, shows the bill for adding
two numbers.

\begin{definition}[Leveled volatile substrate]\label{def:substrate}
\textbf{Storage.} Storage is in \emph{cells} of one bit, enumerated
within each level in a fixed order that groups them into
\emph{words} of $w$ consecutive cells, the groups fixed (words are
aligned), and may leave cells outside every word; such a cell never holds an input bit (whenever it is live it is computed). A
substrate has \emph{working} storage levels $0,\dots,L$ with
capacities $S_\ell$ (bits), access latencies $\tau_0 \leq \dots \leq
\tau_L$ (only $\tau_0$, the least, is used below), and \emph{hold powers} $p_\ell \geq 0$ (the least energy
per resident bit per unit time on the level; where a conveyor, below,
runs, the level may have one hold power while it runs and another
while it rests, and $p_{\min}$ ranges over both), and an
\emph{environment} level $L{+}1$ of
unbounded capacity and hold power $p_{L+1} = 0$, in two regions: an
\emph{input} region, which holds the input at the start and is
read-only, and an \emph{output} region, which must hold the output at
the end and is write-only.
\par\textbf{Fares.} Moving a bit across boundary $\ell$
(between levels $\ell$ and $\ell{+}1$) costs \emph{fare} $c_\ell$
(energy per bit); a bit crosses boundary $L$ downward only from the
input region, and only as part of its whole input word, which lands
in one word of level $L$ (the input is read, and may be re-read, at
random access, one word at a time; the fare is still per bit), and
upward only from computed cells and only into the output region,
from which nothing is ever read, and there too as a unit: the $w$
cells of one computed aligned word cross within one round into one
word of the output region, which is written once (an output bit is
what a fire wrote or what the program held, so an input word reaches
the output only through a fire that read it; and an output word is
an item, which is what the bounds we import count): the environment
is an input tape and an output tape, as in the sequential models we
import, not a store. The levels are written as a line for the
notation's sake; a substrate may instead have a graph of boundaries,
the environment adjacent to one designated working level, a fare and
a traffic count per edge, a move crossing one edge and charged in
either direction, and the sums of \eqref{eq:law} and the minimum
$p_{\min}$ below running over edges and levels; nothing below uses
the line except the two-level statements, which name their two
levels.
\par\textbf{Operations.} Operations fire at $s$
\emph{stations}, each at a working level, reading at most two words
resident at its level and writing bits at that level, in any number
(no result below depends on a bound), as a function of what it reads
(its operation), with per-operation energies. Lemma~\ref{lem:serve},
Definition~\ref{def:atomic}, Theorem~\ref{thm:meeting}, and
Theorem~\ref{thm:hk} place their stations at level 0, as two-level
statements do; nothing else below uses the level of a station.
\par\textbf{Steps.} A bit is \emph{live} at an instant if it
is held at a working level; the live cells with their levels and
contents, the environment's contents, and the time are the
\emph{state}, and a level never holds more than its capacity. A
\emph{step} is a \emph{hold} (time advances by the step's duration $d
> 0$; a level may carry its contents along a \emph{conveyor}, a
fixed permutation of its words, carrying each aligned word whole
with the order of its cells preserved, depending neither on the input
nor on the state, applied once per unit of time, a hold on such a
level lasting a whole number of units, as a shift register or a
rotating ring carries them; nothing is read, written, created, or
destroyed), a \emph{move} (one bit crosses one boundary into a
\emph{free} cell, one that is not live, at the far level, which must
have one within its capacity; the source
keeps its contents, so a move between working levels is a copy and
the source stays live until a fire frees it, a move from the input
region leaves the region as it was, and a move into the output
region writes it), or a \emph{fire} (one operation at one station; it may
also free live bits, whether or not it read or wrote them, which
cease to be live). A \emph{round} is a maximal
sequence of moves and fires between two holds. An \emph{execution} is
a finite sequence of states from an initial state, consecutive states
related by a step; its last state is its \emph{final} state, and an
execution may end without a final hold, so that the bits live in its
last round pay nothing, which only lowers the action. (In the
idiom of temporal logic an execution is a behavior that stutters after
its last step; holds are not stuttering steps, since they advance
time.)
\par\textbf{Action and synchrony.} A level's live count is constant
during holds, a conveyor
relocating live bits without changing their number, so the per-level
\emph{action} is the sum over holds $\act_\ell = \sum_{\text{holds}}
\mathrm{live}_\ell \cdot d = \int \mathrm{live}_\ell(t)\,dt$
(bit-seconds), and the \emph{traffic} $\traf_\ell$ is the number of
bits crossing boundary $\ell$. An execution is \emph{synchronous with
cycle $\Delta$ and width $q$} if every hold lasts at least $\Delta$
and every round contains at most $q$ fires; a clocked machine with $q$
stations and cycle time $\Delta$ produces such executions.

\par\textbf{Machines.} A \emph{machine} $M$ chooses the steps. Call
a live cell \emph{raw} if it holds an input bit not yet overwritten
by a fire and \emph{computed} otherwise. Raw words are units: an
input word enters a working level whole, and its cells are
thereafter moved between working levels, into one aligned word,
overwritten, and freed together, within one round (a fire that
overwrites a cell of a raw word overwrites every cell of it, and the
word's cells are thereafter computed, holding what the fire wrote; a
fire that reads a raw word leaves it raw; there is no relabeling, so
an input bit reaches the computed string only through a fire that
read its word), and a conveyor carries words whole, so a word of any
working level holds bits of at most one input word. The initial
contents of the computed cells are the
machine's \emph{program} and do not depend on the input. (i) The
next step, including which operation fires at which station, which
words it reads, where a move goes, whether the round ends, whether
the execution ends, the
duration of the hold that follows, and whether each level's conveyor
runs during it, is a function of the contents of
the computed cells, by which we mean the assignment of contents to
the computed set, the membership of that set included (fires grow it
by overwriting raw cells); and the assignment is
\emph{self-describing}: two computed assignments reachable in $M$'s
executions, on any inputs, whose contents agree when read in the
fixed order of the cells are the same assignment, so control factors
through the \emph{computed string} and which cells are occupied is
information the string already carries, never a register beside it.
(On a level whose conveyor runs this has a consequence: a conveyor
carrying $k$ live bits among $S$ cells would take one string through
$S$ distinct assignments, so such a level records its occupancy in
live cells; the ring does so with one live cell per slot, occupied or
not, its occupancy flop, whose rent on the empty slots is the idle
floor of Section~\ref{sec:instrument} and on an occupied slot is
inside the clock term of Table~\ref{tab:fit}.)
Raw cells do not steer control, which depends on
data only through fires, a move or a hold reads nothing, and neither
physical time nor any count of steps, rounds, or holds is an
argument. A machine that schedules by count or by clock keeps the
count in computed cells, where it pays rent like any other live
bits, and a machine that would record a bit in \emph{which} cell it
writes finds the choice in its string, where it pays the same; two
round starts with the same computed string begin identically and
diverge only through what their fires read, and since a conveyor's
permutation is fixed while
the hold's duration and whether the conveyor runs are functions of
the string, the computed
string at a hold's end is a function of the string at its start, so
a statement about the string during a hold may read it at either
end. (ii) Throughout every hold, the set of live cells,
with their levels, and the assignment of input words to raw words,
are likewise a function of the computed string: occupancy and
placement are recorded in live cells (a ring slot's occupancy flop
is such a cell), never in the choice of which cells are occupied. On
each input $x$, placed in the input region, $M$ produces one
execution $M(x)$; $M$ \emph{computes} $f$ if for every $x$ the final
output region of $M(x)$ holds $f(x)$. A randomized machine reads its
random bits from a region of the input, once; the coin region is
exempt from the input's re-readability, and a coin word moved to a
working level is a raw word like any other (Lemma~\ref{lem:sim}
queries it when a fire reads it). It computes $f$ with the
probability, over those bits, that it does so. (A machine whose
inputs are resident on a working level from the start, as a ring's
seeded packets are, is the special case whose initial state holds their payloads, raw,
at their seats, a packet's header cells being computed; its
instruction packets are its program, and its rotation is its level's
conveyor.) The
theorems below about \emph{functions} (Lemma~\ref{lem:sim},
Theorem~\ref{thm:rentbridge}, Corollary~\ref{cor:sortjoules},
Lemma~\ref{lem:prom}, Corollary~\ref{cor:mhf},
Theorem~\ref{thm:uncond}) quantify over machines; the others
(Lemma~\ref{lem:serve}, Theorem~\ref{thm:meeting} and its corollaries,
Theorems~\ref{thm:action} and~\ref{thm:halforbit}) are statements
about single executions and need none.
\end{definition}

\begin{axiom}[Volatility]\label{ax:vol}
$p_\ell > 0$ for every level fast enough to feed computation. Device
physics couples speed to leak: fast transistors
leak~\cite{RoyMM03,KimAustin03}, DRAM refreshes~\cite{LiuJVM12}, and
the instrument's rings rotate; $p$ falls steeply as $\tau$ grows,
and levels with $p \approx 0$ (flash, disk) have large $\tau$ and large
write fare.
\end{axiom}

\begin{axiom}[Embedding]\label{ax:embed}
Every physical execution of an algorithm is an execution of a machine
in the sense of Definition~\ref{def:substrate}: at every instant its
live bits reside at levels, its movements cross boundaries or ride a
conveyor, its operations fire, its control is steered by data only
through the words its fires read, and its occupancy is recorded in
live cells. In particular its behavior after any hold is a function
of the contents of its live cells as that hold ends and of the input
bits it reads afterwards: the live cells are the whole working state, a
sequencer's counters and timers included, which is why
Definition~\ref{def:substrate} gives control no argument beyond
them; a clocked machine's schedule is the contents of the cells that
hold its count.
This is a thesis in the Church--Turing sense~\cite{Valiant24}, not a
theorem; everything we prove below assumes it.
\end{axiom}

\begin{axiom}[Latency]\label{ax:lat}
A station fires at most once in any interval shorter than $\tau_0$,
and a fire that reads a bit written or moved into the station's level
at time $s$ (a bit resident there from the start counts as moved in
at time $0$) occurs at time $\geq s + \tau_0$. (The level-0 access
time, the least, bounds both a station's rate and a dependence's
spacing; the slower levels' latencies enter only through their fares
and are not used below.)
\end{axiom}

\noindent
\textbf{The action law.} For any execution,
\begin{equation}\label{eq:law}
  E \;\geq\; \sum_{\text{ops}} e_{\mathrm{op}}
  \;+\; \sum_\ell p_\ell\,\act_\ell
  \;+\; \sum_\ell c_\ell\,\traf_\ell .
\end{equation}
Read literally, with $p_\ell$, $c_\ell$, $e_{\mathrm{op}}$ the infima of the physical costs of holding, crossing, and firing, dissipated in disjoint physical events, \eqref{eq:law} is true by definition. Its empirical content is Axiom~\ref{ax:vol}, that the infima at the fast levels are strictly positive, together with the magnitudes we extract in Section~\ref{sec:instrument}, which put the floors at the picojoule scale. Its mathematical content is that lower bounds on $\act$ and $\traf$ from the \emph{structure} of a computation transfer to joules, and that there exists a machine (Section~\ref{sec:instrument}) on which we extract $\geq$ as $\approx$, so the constants are real. All constants are at a fixed operating point (voltage, temperature,
frequency): leakage moves exponentially with the first
two~\cite{RoyMM03,KimAustin03}, and voltage or frequency scaling moves
the constants, not the form. The three terms, with the erasure toll of
Remark~\ref{rem:erasure}, are four literatures (Table~\ref{tab:four}).

\begin{example}[Adding two words]\label{ex:add}
Take one working level of capacity $S_0$ bits, hold power $p_0$, and
fare $c_0$ at its boundary with the environment; one station that
adds; and an input of two $w$-bit words $x_1, x_2$. One execution: in
the first round, move $x_1$ and $x_2$ into the level, each landing in
one aligned word ($2w$ crossings, fare $2w\,c_0$); hold for a time
$\Delta \geq \tau_0$, during which $2w$ bits are live (rent
$2w\,p_0\Delta$); in the second round, fire the adder on the two raw
words, which writes the $w$-bit sum into one computed word and frees
$x_1$ and $x_2$, and move the sum, as one word, to the output region
($w$ crossings, fare $w\,c_0$). The execution ends without a final hold.
Its bill by \eqref{eq:law} is
\[
  E \;\geq\; e_{\mathrm{add}} + 3w\,c_0 + 2w\,p_0\Delta :
\]
the addition, the three word fares, and the rent of holding the
operands for the one cycle that Axiom~\ref{ax:lat} demands. A machine
that held the operands for ten cycles before adding would pay ten
times the rent for the same answer, and one that added a thousand
pairs while keeping every sum live would pay rent on all of them for
as long as it kept them. That is the whole content of the law: the
arithmetic is the small part of the bill.
\end{example}

\begin{remark}[Erasure]\label{rem:erasure}
Landauer's toll, $kT\ln 2$ per erased bit~\cite{Landauer61}, sits eight to nine orders below the
operation and fare terms on current substrates; rent, a power,
compares per use, $p_0$ times the gap, and stays orders above the
toll at every gap we price, from two to three for an SRAM bit at
nanosecond gaps to six or seven at DRAM refresh or on our
instrument. We do not add it to \eqref{eq:law}: every measured
operation energy already contains its own Landauer floor, so a
separate additive term would charge the same physics twice. It reappears in Section~\ref{sec:falsify} as the regime that reversible logic escapes to.
\end{remark}

\begin{table}[t]
\caption{One law, four literatures.}
\label{tab:four}
\small
\begin{tabular}{@{}llll@{}}
\toprule
term & prices & theory & unit \\
\midrule
ops & doing & circuit complexity & pJ/op \\
$\sum p_\ell \act_\ell$ & keeping & cumulative memory \cite{AlwenSerbinenko15,BeameKornerup23} & (W/bit)$\cdot$s $=$ J/bit \\
$\sum c_\ell \traf_\ell$ & moving & I/O complexity \cite{HongKung81,AggarwalVitter88,Kung86} & pJ/bit \\
erasures & forgetting & Landauer \cite{Landauer61,Bennett73,DemaineLMT16} & $kT\ln 2$ \\
\bottomrule
\end{tabular}
\end{table}

\begin{remark}[Levels discretize distance]\label{rem:distance}
On chip the fare of a bit is a function of the distance it
travels~\cite{KecklerDally11}, and Bilardi and Preparata make that a
consequence of the speed of light~\cite{BilardiPreparata95}. The
leveled model prices distance by the boundaries crossed, which is how
hierarchies are built and billed; the ring's window $W$
(Section~\ref{sec:structure}) is the same geometry made explicit: reach
is bounded, and what is out of reach must move.
\end{remark}

A \emph{value} is an annotation of an execution, not an object of
Definition~\ref{def:substrate}: a labeling of physical bits by (value,
bit position) over time, and of firings as copies or route firings.
The lemma's hypotheses constrain the annotation, and each application
(Definition~\ref{def:atomic}, Theorem~\ref{thm:action}) supplies one.
Our workhorse for lower bounds is a serving lemma: there are only three
ways for a value to be present when an operation needs it.

\begin{lemma}[Serving lemma: rent or fare or recompute]\label{lem:serve}
Fix an execution in which (i) distinct live values occupy pairwise
disjoint physical bits, and (ii) for each value $v$, every level-0
firing that writes bits of $v$, other than the one creating $v$,
writes bits of no
other value and is either a \emph{copy} (it reads the same bits of $v$
from a level-0-resident copy) or belongs to a \emph{recomputation
route} of $v$: the non-copy firings are partitioned into routes; a
route writes bits of one value, all of them in its last firing, reads
no bit of that value, and dissipates at least $\gamma_{\mathrm{rc}}$
per bit it writes, counted per write (a bit written twice, by two
routes, is paid twice); bits a route writes that belong to no value
(its intermediates) are unlabeled and free of (i) and (ii). (Whenever
each writing firing can be assigned the firings that computed what it
writes, as it can for every route class we price, the non-copy
firings admit such a partition.) For each value
$v$ of footprint $\beta_v$ bits let $t^v_0$ be its creation (the first
instant a bit of $v$ is at a working level; the annotation may
relabel resident bits to a new value, which is then created at that
instant) and $t^v_1 \leq \dots \leq t^v_{k_v}$ the times of the
level-0 operations reading it other than copy firings (a copy serves
a later use and is not one; counting it as a use would only raise
the sum below, since $\min(a, c) + \min(b, c) \geq \min(a + b, c)$),
and call bit $b$ of $v$ \emph{covered}
at use $i$ if at every instant of $[t^v_{i-1}, t^v_i]$ some level-0
copy of $b$ is resident. Then
\[
  E \;\geq\; \sum_{v}\; \sum_{i=1}^{k_v} \sum_{b} \mathrm{ch}(v,i,b),
\]
where $\mathrm{ch}(v,i,b) = \min( p_0\,(t^v_i - t^v_{i-1}),\, c_0,\,
\gamma_{\mathrm{rc}})$ if $b$ is covered at use $i$ and $\min(c_0,
\gamma_{\mathrm{rc}})$ otherwise; hence in particular $E \geq \sum_v \sum_i \beta_v \min(p_0 (t^v_i -
t^v_{i-1}), c_0, \gamma_{\mathrm{rc}})$; and the charged energy is
disjoint from the energies of the reading operations that are not
route firings of another value.
\end{lemma}

\begin{proof}
A potential argument; it is also the inductive invariant that TLC
checks and, for atomic values, TLAPS proves in Appendix~\ref{app:tla}. Write $\mathrm{last}(v,b)$ for the
time of the previous use of bit $b$ of $v$ (its creation before the
first use), $\mathrm{Ch}$ for the sum of the charges of the uses so
far, and at each instant $t$ define the \emph{pending} charge
\[
  \Phi \;=\; \sum_{(v,b)\ \text{resident}} \pi(v,b),
\]
with $\pi(v,b) = \min(p_0 (t - \mathrm{last}(v,b)), c_0,
\gamma_{\mathrm{rc}})$ if $b$ has been resident since
$\mathrm{last}(v,b)$ and $\pi(v,b) = \min(c_0, \gamma_{\mathrm{rc}})$
otherwise, where ``resident'' means that some level-0 copy of $b$
exists and ``resident since $\mathrm{last}$'' that one has existed at
every instant since; and write $R$ for the energy so far dissipated by
the routes still in progress (a route is in progress from its first
firing until its last, which by (ii) is the one that writes). Initially $E =
\mathrm{Ch} = \Phi = R = 0$. We check that $E - \mathrm{Ch} - \Phi -
R$ never decreases along a step; at the end no route is in progress,
so $E \geq \mathrm{Ch} + \Phi \geq \mathrm{Ch}$.
\begin{enumerate}
\item[1.] \emph{Hold of duration $d$.} By (i) distinct resident pairs
  $(v,b)$ occupy distinct physical bits, so $E$ rises by at least
  $p_0 d$ per resident pair; each pending term rises by at most $p_0 d$.
\item[2.] \emph{Move of a copy of $b$ into level 0.} $E$ rises by
  $c_0$. If $(v,b)$ was not resident, a term appears; it is
  $\min(c_0, \gamma_{\mathrm{rc}}) \leq c_0$, because a pair that is
  not resident now was not resident at every instant since its last
  use. Otherwise nothing changes.
\item[3.] \emph{Route firing.} A firing of a route other than its
  last writes no bit of $v$ and raises $E$ and $R$ by the same amount.
  The last firing completes the route: $R$ falls by the route's energy
  so far, $S$, $E$ rises by the firing's own energy $e_f$, and by (ii)
  $S + e_f \geq \gamma_{\mathrm{rc}}$ per bit written, while each pair
  made resident adds a term $\leq \gamma_{\mathrm{rc}}$, as in step 2
  (a pair already resident that is rewritten adds nothing). The firing
  that creates $v$ is exempt from (ii), and the
  pairs it makes resident have $\mathrm{last} = t$ and term $0$; a
  relabeling is a creation without a firing.
\item[4.] \emph{Copy firing.} It reads a resident copy, so $(v,b)$ was
  already resident and no term appears; $E$ does not fall.
\item[5.] \emph{Eviction or freeing.} Terms vanish; $E$ is unchanged.
  A pair that loses its last copy is thereafter ``not resident since
  last'' until its next use.
\item[6.] \emph{Use at time $t$.} The read requires a resident copy, so
  $\pi(v,b)$ is present, and it equals $\mathrm{ch}(v,i,b)$ exactly:
  ``resident since last'' is ``covered''. The charge moves from $\Phi$
  into $\mathrm{Ch}$, $\mathrm{last}$ is reset, and the term becomes
  $0$. The reading operation's own energy is never used unless the
  operation is a route firing of another value, which is the
  disjointness claim.
\end{enumerate}
The weak form follows because the uncovered charge dominates the
three-way minimum.
\end{proof}

\noindent\textbf{In words.} Suppose an operation at time $t_2$ needs a
value that was last used at time $t_1$. Between the two uses one of
three things happened to each of its bits: it stayed in fast storage
the whole time, paying rent $p_0(t_2 - t_1)$; it was brought back from
a slower level, paying the fare $c_0$; or it was computed again,
paying at least $\gamma_{\mathrm{rc}}$. The lemma says that the
cheapest of the three is a floor on what that use cost, that the
floors of different uses are paid by different physical events, and
so that they add up. The phone number of the introduction is the case
of one value used twice.

The lemma prices a tradeoff, not a verdict: a computation may pay any of
the three currencies per use, and the architectures of Section~\ref{sec:meeting} differ exactly in which they choose. The floor is
the min. Two of its three constants are substrate constants; the third
is not: $\gamma_{\mathrm{rc}}$ is a property of the routes the
execution uses, and the lemma is exactly as strong as the route
hypothesis one is willing to state. Unconditionally, a route must at
least write its output, so the level-0 per-bit write energy
($\sim$0.1--1\,pJ/bit) is a floor beneath any $\gamma_{\mathrm{rc}}$;
the $10^3$\,pJ/bit of Table~\ref{tab:attention} is the cost of one
route class, re-projection, and is a hypothesis about the execution
(Definition~\ref{def:atomic}), not a law of the substrate; no lower bound of that size exists for computing a linear map. Hypothesis (i) is
real, not bookkeeping: an execution storing $x$ and $x \oplus y$ holds
both values in fewer bits than their combined footprints. It is
licensed exactly when the live values are \emph{jointly} incompressible, the regime of every application below, and Corollary~\ref{cor:aggregate} gives the form that needs only the joint
state's incompressibility.

\section{Bridges: Old Bounds Become Joule Floors}
\label{sec:bridges}

The bridge has two steps. First, any machine of
Definition~\ref{def:substrate} can be redrawn as a \emph{branching
program}, a decision tree that reads one input word per level, so that
the tree's memory at each level is at most the machine's live bits
(Lemma~\ref{lem:sim}). Second, known lower bounds on the memory of
such trees, summed over their levels, become lower bounds on live bits
summed over holds, which is rent (Theorem~\ref{thm:rentbridge}). A
reader who takes the simulation on trust may skip to
Theorem~\ref{thm:rentbridge}.

\begin{lemma}[Simulation]\label{lem:sim}
Let $M$ be a machine computing $f$ (Definition~\ref{def:substrate})
whose executions are synchronous with cycle $\Delta$ and width $q$.
Call a fire \emph{querying} if it reads a raw word, and let $F(x)$
be the number of rounds of $M(x)$ containing a querying fire. Then
there is a leveled branching program computing $f$ with $M$'s
success probability ($1$ for a deterministic $M$), at most one query
per level, whose nodes are encoded by pairs of strings: the
machine's computed string at the start of the node's round, and the
\emph{answers}, the contents of the distinct raw words read by fires
so far in that round, in the order of their first reading. The nodes
are distinct as pairs within each level apart from one designated
empty node per level, distinguished from every other node, from
which no output is produced; the root, alone at its level, is
encoded as the empty pair. On every
input $x$ the run produces its last output within $2qF(x)$ levels
(an output on an edge counts at the edge's source level) and, once
$M(x)$ has halted, visits only empty nodes; its node lengths sum to
at most $2q \sum_t |\mathrm{live}_t(x)| \leq 2q\,\act/\Delta$, where
$\mathrm{live}_t(x)$ is the live set during the $t$-th hold and
$\act = \sum_\ell \act_\ell$; hence its cumulative memory in the
sense of~\cite{BeameKornerup23}, the sum over levels of the
logarithm of the width, is at most $2q \sum_t (\max_x
|\mathrm{live}_t(x)| + \lceil \log_2 (\max_x |\mathrm{live}_t(x)| +
2) \rceil + 1)$, the logarithm covering the pair's split.
\end{lemma}

\begin{proof}
By (i) of the machine definition the next step is a function of the
computed contents, which the initial state fixes independently of the
input, and by Axiom~\ref{ax:lat} no fire reads a bit written or moved
to its level in its own round; so between two queries the run evolves
as a function of the computed string at the start of the current
round and of the raw words read in it so far, the same for any two
inputs that agree on them. Level the program by queries, one per
distinct raw word a fire reads in a round, and compose the evolution
through rounds without queries, moves, holds, fires that read no raw
word, and reads of a word already read in the round, whose value the
answers hold, into the edges: a fire reads at most two words, each
holding bits of at most one input word since raw words are units, so
a round contributes at most $2q$ levels, and a round without a
querying fire contributes none. The pair encoding a node
determines, by self-description, (i), and (ii), the machine's
control, which cells are live, and which input word each raw word
holds (the computed string determines its assignment, the assignment
the raw placement), so the node is well defined and each query is a
definite input position: the string determines the round's first
query, the first answer with the string its second, and so on, which
is also how the answers are parsed. A word read in an earlier round
is queried afresh when a fire reads it again: the program is not
read-once, and a node need not remember an answer whose trace has
left the computed contents. Two distinct nodes at one level differ
in their pair, which is their encoding (distinct assignments have
distinct strings, by self-description); the root is the only node at
level $0$, since the computed string at the start of the first
querying round is the same on every input (control up to the first
query is a function of the program, and a randomized machine's coins
are raw words queried like any other), so we encode it as the empty
pair, which lowers the length sum and keeps the nodes distinct within
their levels; and both parts of every other node's pair were resident
throughout the hold preceding the round: the round's own writes are
invisible to its own fires by Axiom~\ref{ax:lat} and never enter the
encoding, a conveyor relocates without reading or writing, the
computed cells are disjoint from the raw cells, and the answers list
distinct resident raw words, so the encoding has at most
$|\mathrm{live}_t(x)|$ bits. Every output is computed by fires, since
an output word enters the output region only from one computed word,
whole and within one round, into a word written once; so the program
emits it as an item, the position and the word, from the values the
pair determines, on the edge along which the moves that carry it
into the output region occur, counting at that edge's source level,
and its outputs are the items the bounds we cite count, each
produced once. The end of a round is a
control decision the pair determines, and so is halting: after
$M(x)$ halts, a run passes through the designated empty node of each
later level, so the program is leveled and the padding adds neither
queries nor node length. The program never reads the output region,
which by Definition~\ref{def:substrate} the machine cannot either;
so it computes $f$ whenever $M$ does. On $x$ the last output falls
within $2qF(x)$ levels, the node lengths sum to at most $2q\sum_t
|\mathrm{live}_t(x)|$, and since hold $t$ lasts $d_t \geq \Delta$,
$\act = \sum_t |\mathrm{live}_t|\, d_t \geq \Delta \sum_t
|\mathrm{live}_t|$. A node's pair injects into one string, its parts
concatenated behind the length of the first, so with $K = \max_x
|\mathrm{live}_t(x)|$ a level holds at most $2^{K + \lceil \log_2
(K+2) \rceil + 1}$ nodes, the empty node included, whose logarithm
is at most $K + \lceil \log_2 (K + 2) \rceil + 1$.
\end{proof}

The synchrony hypothesis is the machine's to choose: holds of
$\tau_0/1000$ with a fire every thousandth round are synchronous with
cycle $\tau_0/1000$. For one station Axiom~\ref{ax:lat} alone spaces
the queries, and the lemma has a form with no hypothesis on the
holds.

\begin{lemma}[Simulation, one station]\label{lem:simtau}
Let $M$ have one station and compute $f$, with no hypothesis on its
holds. Then there is a leveled branching program as in
Lemma~\ref{lem:sim}, with the same encoding and distinctness
properties, in which on every input $x$ the last output falls within
$2F(x)$ levels and the node lengths sum to at most $3\act/\tau_0$;
moreover $\act \geq w\,\tau_0\,F(x)$.
\end{lemma}

\begin{proof}
A round is instantaneous and a station fires at most once in any
interval shorter than $\tau_0$ (Axiom~\ref{ax:lat}), so one station
fires at most once per round and consecutive fires, querying fires
among them, are separated by a time $\geq \tau_0$ spanned by holds;
a fire reads at most two raw words. Between querying fires $k$ and
$k{+}1$ (before the first, from the start; that interval too lasts
$\geq \tau_0$, the word read at the first fire having crossed at
some $s \geq 0$ and being read at $\geq s + \tau_0$) choose a hold
$h_k$ at
which the live count is least over the interval's holds. Encode the
first node of fire $k{+}1$ by the computed string at $h_k$ and, if
the fire reads two raw words, the second by that string with the
first word's contents as the answer; before the first query the
string is the same on every input, and the root is the empty pair
as before. From $h_k$ up to the fire the run is a function of the
string alone, raw cells steering nothing and every fire in between
reading no raw word, and past the fire it is a function of the
pair, so the evolution folds into the edges as in
Lemma~\ref{lem:sim}, and the nodes are distinct as pairs within
their levels by self-description. The first node has length at most
the least live count over the interval and the second at most $w$
more. The raw word the fire reads last crossed into the station's
level at some time $s$ and is read at time $\geq s + \tau_0$, so its
$w$ cells are live through a stretch of length $\geq \tau_0$ inside
the interval, from $s$ or from the interval's start, whichever is
later, to the fire: $\act \geq w\tau_0$ per querying fire, and $\act
\geq \tau_0 \sum_k \min_k |\mathrm{live}|$ since each interval lasts
$\geq \tau_0$ at a live count at least its minimum and the intervals
are disjoint. The node lengths therefore sum to at most $2 \sum_k
\min_k |\mathrm{live}| + w F(x) \leq 3\act/\tau_0$.
\end{proof}

\begin{theorem}[Rent bridge]\label{thm:rentbridge}
Let $f$, a distribution $\mu$ on its inputs, and $\delta : (0, 1]
\to (0, 1]$ be such that every leveled branching program computing
$f$ with success probability $\gamma \geq 1/\mathrm{poly}$ over $\mu$
and its coins, whose nodes carry encodings, pairs of strings as in
Lemma~\ref{lem:sim}, distinct within each level apart from at most
one outputless empty node per level, has a set of inputs of
$\mu$-probability $\geq \delta(\gamma)$ on each of which the run
produces its last output at level $\geq T_f$ or has node lengths
summing to $\geq M_f$. Then on any substrate whose working levels are
all volatile, with $p_{\min} = \min_{\ell \leq L} p_\ell > 0$, every
machine computing $f$ with success probability $\gamma$ pays, in
synchronous executions of cycle $\Delta$ and width $q$, on a set of
inputs (of input and coin pairs, for a randomized machine) of
probability $\geq \delta(\gamma)$,
\[
  E \;\geq\; p_{\min}\,\Delta\,\min\!\Big(\frac{M_f}{2q},\;
  \frac{w\,T_f}{2q}\Big),
\]
however the input is held: the environment holds it rent-free and it
may be re-read at will. If $M$ has one station, then in every
execution, with no synchrony hypothesis,
\[
  E \;\geq\; p_{\min}\,\tau_0\,\min\!\Big(\frac{M_f}{3},\;
  \frac{w\,T_f}{2}\Big)
\]
on a set of inputs of probability $\geq \delta(\gamma)$.
\end{theorem}

\begin{proof}
By Lemma~\ref{lem:sim} the simulating program produces its last
output on $x$ within $2qF(x)$ levels, $F(x)$ the number of rounds
containing a querying fire, with node lengths summing to at most
$2q\sum_t|\mathrm{live}_t(x)|$. By Axiom~\ref{ax:lat} a querying
fire's raw operand entered the fire's level at least $\tau_0$ before
it, hence in an earlier round, so that word, a unit of $w$ raw cells, is
resident throughout the preceding hold, and every round counted by
$F(x)$ has $|\mathrm{live}_t(x)| \geq w$:
$\sum_t|\mathrm{live}_t(x)| \geq w F(x)$. On each input of the
$\delta$-set, either the node lengths sum to $\geq M_f$, and then
$\sum_t |\mathrm{live}_t(x)| \geq M_f/(2q)$, or the last output falls
at level $\geq T_f$, and then $F(x) \geq T_f/(2q)$ and
$\sum_t|\mathrm{live}_t(x)| \geq w\,T_f/(2q)$. By Axiom~\ref{ax:vol}
every live bit pays at least $p_{\min}$ per unit time, and $\act \geq
\Delta\sum_t|\mathrm{live}_t|$. For one station use
Lemma~\ref{lem:simtau} instead: on the $\delta$-set either the node
lengths sum to $\geq M_f$, so $\act \geq \tau_0 M_f/3$, or the last
output falls at level $\geq T_f$, so $F(x) \geq T_f/2$ and $\act
\geq w\tau_0 T_f/2$.
\end{proof}

\begin{corollary}\label{cor:sortjoules}
Beame--Kornerup~\cite{BeameKornerup23} prove (their Theorem 3.3) that
any branching program computing $\mathrm{Sort}_{n,n^3}$, $n$ items of
$3\log_2 n$ bits, with probability $\geq n^{-O(1)}$ has time
$\Omega(n^2/\log^2 n)$ or cumulative memory $\Omega(n^2/\log n)$, the
latter being the sum over levels of the logarithm of the width; the
disjunction collapses on random-access machines, which hold a word at
every step. Their proof gives the per-input form
Theorem~\ref{thm:rentbridge} needs, with $\delta = 15/16$ for a
program correct on every input (Proposition~\ref{prop:perinput},
Appendix~\ref{app:perinput}: the union bound runs over the at most
$(s{+}1) 2^s$ pair-encoded nodes of total length $s$ on each
level of a block instead of over the
level's width, and the time branch is read on each run's own output
span). Hence on
any substrate whose working levels are volatile, a deterministic
machine with one station sorting $n$ items, one per input word ($w
\geq 3 \log_2 n$), pays rent on its
working set, whatever its holds (a randomized one, correct with
probability $\gamma \geq 1/\mathrm{poly}$, pays the same on input
and coin pairs of probability $> \tfrac{15}{16}\gamma$, by the second
clause of Proposition~\ref{prop:perinput}), on at least fifteen sixteenths of the inputs of the
form $(\langle y_i, i\rangle)_{i \leq n}$ with $y$ a list of $n$
distinct integers from $[n^2]$ and $\langle \cdot, \cdot \rangle$
the pairing $[n^2] \times [n] \to [n^3]$ through which
\cite{BeameKornerup23} pass from ranking to sorting, under the
uniform distribution on those lists,
\[
  E \;\geq\; p_{\min}\,\tau_0\,\Omega(n^2/\log n) ,
\]
$\tau_0$ the access time of its fastest level (Axiom~\ref{ax:lat}; in
synchronous executions of cycle $\Delta$ and width $q$ the first
form of Theorem~\ref{thm:rentbridge} gives $p_{\min}\Delta\,\Omega(n^2/(q \log n))$): an energy floor independent of the
algorithm, at the substrate's scale rather than Landauer's, on the
\emph{working set}, which the input, held rent-free and re-readable,
does not pay (a word-granularity floor would need an argument not
in~\cite{BeameKornerup23}). The floor is structural before it is
practical: at SRAM-class rent and nanosecond access times, $p_{\min}\tau_0
\approx 10^{-18}$\,J per bit-step, it passes the sorter's own
operation cost, picojoules times $n \log n$, only near $n \approx
10^9$; the meeting bound of Section~\ref{sec:meeting} is where the
law's floors reach engineering scales. Their matrix-multiplication
bound (their Theorem 6.9: cumulative memory $\Omega(n^6 \log d/T)$
over a field of size $d$, in a disjunction with a time bound)
transfers by Lemma~\ref{lem:sim} only to the width form, $\sum_t
\OO(\max_x |\mathrm{live}_t(x)| + 1)$, the profile of live counts over
inputs; the width form bounds no single run's action, so we claim no
energy floor for matrix multiplication, and we have not carried its
proof to the per-input form.
\end{corollary}

\noindent\textbf{In words.} However a one-station machine sorts, and
however it paces itself, on most inputs it must keep at least a
constant times $n^2/\log n$ bit-steps of working state alive, a step
being one access time of its fastest level, and each bit-step costs
at least $p_{\min}\tau_0$
joules. Reading the input again is free; remembering what has been
learned from it is not. For a million items on SRAM the floor is
still below the cost of the comparisons; the point is that it exists,
that no algorithm removes it, and that the same argument gives large floors elsewhere
(Corollary~\ref{cor:mhf} and Section~\ref{sec:meeting}).

\begin{remark}[Where the input lives]\label{rem:input}
The environment level carries weight. Were the $n$ input words
resident on a volatile working level from the start, as they are on a
ring, then for \emph{every} function that depends on all of them a
one-station execution would pay rent $\geq p_{\min}\tau_0\, w\,
n^2/4$ with no synchrony hypothesis: one station's fires are $\geq
\tau_0$ apart (Axiom~\ref{ax:lat}) and each reads at most two words,
so the $k$-th word read waits at least $\lceil k/2 \rceil \tau_0$, and
$\sum_k \lceil k/2 \rceil \geq n^2/4$. That floor, which is the seed-residency argument
of Theorem~\ref{thm:action}, would dominate
Corollary~\ref{cor:sortjoules} by the factor $w$ and the polylogarithm
and would owe nothing to sorting; for matrix multiplication it would be
$\Omega(n^4)$ against $\Omega(n^6/T) \leq n^4$ for every feasible $T
\geq n^2$. Hong--Kung's blue pebbles are unbounded and unpriced for the
same reason: the theorem is about the working set, and the bridge has
content only where the input is not charged for waiting.
\end{remark}

\begin{remark}[Parallelism changes the answer, and is priced]
\label{rem:parallel}
The sequential scope matters: a sorting
\emph{network} of depth $\OO(\log n)$ on $n$ wires has cumulative
memory only $\Theta(n \log n)$ word-steps, far below $n^2$.
Parallelism genuinely lowers rent, by provisioning $n$ comparators, the capability currency of Section~\ref{sec:choice}. For
parallel substrates the right residency measure is the
\emph{parallel} cumulative complexity developed for memory-hard
functions~\cite{AlwenSerbinenko15}, and the bridge holds with the
parallel bound in place of $M_f$ once it is stated per input; which
explicit functions keep their cumulative memory
high under parallelism is precisely what that line
establishes~\cite{AlwenBlocki16,ACPRT17}.
\end{remark}

\begin{remark}
On mixed substrates a nonvolatile level ($p \approx 0$) converts rent to
fare at its boundary: parking the live set on disk is free to hold and
costly to touch. The optimization over placements is
Lemma~\ref{lem:serve}; either currency yields a positive floor. The
cryptographic line anticipated the trade: cumulative memory prices the
adversary's silicon~\cite{AlwenSerbinenko15}, bandwidth-hardness prices
its transfers~\cite{RenDevadas17,BlockiRenZhou18}; in our reading both price the same computation in different currencies.
\end{remark}

The cryptographic line also receives a floor back. Its cumulative
bounds are proved in the parallel random-oracle model, and the machine
of Definition~\ref{def:substrate} projects onto that model exactly,
because its control reads nothing but priced state.

\begin{lemma}[Oracle projection]\label{lem:prom}
Extend a machine with an \emph{oracle station} for a function $H$: a
fire that reads one $w$-bit block resident at its level and writes
$H$ of it, at any per-fire energy, with $H$ available only there. Every
execution on input $x$, coin region included, projects to a parallel
random-oracle algorithm~\cite{AlwenSerbinenko15} computing the same
output, with one parallel step per round containing oracle fires, the
step's queries being that round's oracle inputs, and the state
entering it the machine's computed string during the preceding hold
(its raw cells hold input words, supplied by the hardwiring);
$x$ is hardwired into the algorithm's inter-step functions, which the
model does not charge, so one algorithm serves one input. Its
cumulative memory complexity, the sum of the state sizes over
steps~\cite{AlwenSerbinenko15,ACPRT17}, is therefore at most the sum
of $|\mathrm{live}_t|$ over the holds preceding those rounds, hence
at most $\sum_t |\mathrm{live}_t|$, with no dependence on the width
$q$ or on the length of the input. \emph{Grid form.} Fix $\theta \in
[0, \tau_0)$ and call $[g\tau_0 + \theta, (g{+}1)\tau_0 + \theta)$
the $g$-th \emph{window}. The execution also projects to a parallel
random-oracle algorithm $A_\theta$, defined from $M$, $x$, and
$\theta$ alone, with one parallel step per window containing an
oracle fire, the step's queries being the window's oracle inputs and
its state the computed string at the window's first instant; its
cumulative memory complexity is at most $\sum_g
|\mathrm{live}(g\tau_0 + \theta)|$, the counts taken at instants
inside holds, whose average over $\theta$ uniform in $[0, \tau_0)$ is
$\act/\tau_0$, with no hypothesis on the holds and none on the
width.
\end{lemma}

\begin{proof}
By clause (i) the machine's continuation from a hold is a function of
its computed string and of the oracle answers and input words read
afterwards (self-description recovers the assignment from the
string), and by clause (ii) the string determines the live set and
placement; so the computed string during the hold before an
oracle-bearing round, with the input, determines that round's queries
and the whole run to the next such round. The projected algorithm
carries the computed string alone as its state and simulates the
machine between queries, which the oracle model does not charge;
wherever the machine reads an input word, its coins included, the
simulation supplies the word from the algorithm's own description,
into which $x$ is hardwired. The move is sound because the model's
lower bounds hold for every fixed algorithm, with the probability
over $H$ alone, so an algorithm built for one input is as bound as a
uniform one; carrying $x$ in the state instead would add a term
$S\,|x|$ over the $S$ oracle-bearing rounds that the coin region,
whose size Definition~\ref{def:substrate} does not bound, could make
dominant. By Axiom~\ref{ax:lat} a
query's input block is resident through the preceding hold and an
answer is readable only in later rounds, so a round is a parallel
step: its queries do not depend on one another's answers. The step
states are computed strings at distinct holds, each of length at most
the live count there, and a sum over a subset of
holds is at most $\sum_t |\mathrm{live}_t|$.

For the grid form, an oracle answer written at time $s$ is read only
at times $\geq s + \tau_0$ (Axiom~\ref{ax:lat}: read where it was
written it waits $\tau_0$, and moved elsewhere it waits $\tau_0$
after the move), so no fire in a window reads an answer written in
that window, and a window's queries do not depend on one another's
answers: a window is a parallel step. A bit is live from its writing
to its last read, so every answer written before a window's first
instant and read after it lies in the computed string at that
instant (for the finitely many $\theta$ at which that instant is a
round, read the string just before it). Hence the run from the
instant through the window is a function of that string and of the
hardwired input, and of no answer, which gives the step's queries;
and the run on to the next window's first instant is a function of
the same and of the window's answers, which gives the next state.
The state's length is at most the live count at the instant, taken
inside a hold (all but finitely many $\theta$), and live counts are
piecewise constant, so $\int_0^{\tau_0} \sum_g |\mathrm{live}(g
\tau_0 + \theta)|\,d\theta = \int |\mathrm{live}(t)|\,dt = \act$.
\end{proof}

\begin{corollary}[The adversary's joule floor]\label{cor:mhf}
Alwen, Chen, Pietrzak, Reyzin, and Tessaro~\cite{ACPRT17} prove that
every parallel random-oracle algorithm computing
$\mathrm{scrypt}^H_n$ on $w$-bit blocks has cumulative memory
complexity $\Omega(n^2 w)$ with high probability over $H$, at any
parallelism, even amortized over evaluations (the failure
probability grows with the algorithm's number of oracle queries,
which for a machine is the number of oracle fires of the execution,
finite and counted). Hence on any substrate
whose working levels are volatile, every machine computing
$\mathrm{scrypt}^H_n$ through oracle stations, however many, pays in
every execution, whatever its holds and its width,
\[
  E \;\geq\; p_{\min}\,\tau_0\,\Omega(n^2 w)
\]
with high probability over $H$, hence in expectation: the security
parameter of a memory-hard
function is, on volatile hardware, joules per guess, and provisioning
more hashing cores does not lower it, since neither the number of
oracle stations nor the pace of the holds enters the bound.
\end{corollary}

\begin{proof}
By the grid form of Lemma~\ref{lem:prom} each execution projects, its
password and any coins hardwired and for each offset $\theta$, to an
algorithm $A_\theta$ whose cumulative memory is at most $\sum_g
|\mathrm{live}(g\tau_0 + \theta)|$, with the same number of oracle
queries for every $\theta$. A machine computing $\mathrm{scrypt}^H_n$
in the sense of Definition~\ref{def:substrate} is correct on every
input, its coin strings included, so every $A_\theta$ computes
$\mathrm{scrypt}^H_n$, and the cited bound gives, for each $\theta$,
$\Pr_H[\mathrm{cmc}(A_\theta, H) < c\,n^2 w] \leq \varepsilon$ with
$\varepsilon$ negligible. For fixed $H$ and $x$ the map $\theta
\mapsto \mathrm{cmc}(A_\theta, H)$ is piecewise constant, its
breakpoints the round times modulo $\tau_0$, so averaging over
$\theta$ and exchanging the order, $\mathbb{E}_H[\Pr_\theta[\mathrm{cmc}(A_\theta, H) < c\,n^2
w]] \leq \varepsilon$, so with probability $\geq 1 - 2\varepsilon$
over $H$ at least half the offsets have $\mathrm{cmc}(A_\theta, H)
\geq c\,n^2 w$, and then $\act/\tau_0 = \mathbb{E}_\theta[\sum_g
|\mathrm{live}(g\tau_0 + \theta)|] \geq \mathbb{E}_\theta[\mathrm{cmc}(A_\theta,
H)] \geq c\,n^2 w/2$; Axiom~\ref{ax:vol} prices it, at every coin
string. A machine
correct only with probability $\gamma$ over its coins has a coin
string at which its success over $H$ is at least $\gamma$; the cited
bound says that success with small cumulative memory has negligible
probability over $H$, so the projections there pay the floor on the
success event, of probability at least $\gamma$ over $H$ less that
negligible term, and the expected energy at that coin string is at
least $\gamma$ times the floor, up to the same term.
\end{proof}

\noindent\textbf{In words.} A password-cracking chip evaluating scrypt
must keep a large table alive for most of the evaluation, and the bound
is on memory over time, not on time alone; so adding hashing cores
does not reduce the joules per guess, it only spends them faster. At
SRAM-class rent and nanosecond access times ($p_{\min}\tau_0 \approx
10^{-18}$\,J per bit-step) and the common password-hashing parameters
$n = 2^{20}$, $w = 8192$, the floor is of order $10$\,mJ per
evaluation, reading the cited bound's constant, which~\cite{ACPRT17}
do not state, as one, and kilowatt-hours per $10^9$ guesses on the
same reading, before the hashing itself is charged. The cracker's escape to a nonvolatile level
converts the rent to fare at its boundary (Lemma~\ref{lem:serve}), and
pricing that leg is exactly
bandwidth-hardness~\cite{RenDevadas17,BlockiRenZhou18}. The serving
minimum therefore appears at engineering scale twice, on opposite
sides of the ledger: in the datacenter it prices serving a context
(Section~\ref{sec:meeting}); against the adversary it prices a guess.

\begin{theorem}[Fare bridge: matrix multiplication]\label{thm:hk}
On a two-level substrate whose stations reside at a fast level of
capacity $S$ words (the slow level may be the environment), any computation of $n \times n$ matrix
multiplication by its $n^3$ elementary products performs $\Omega(n^3/\sqrt S)$ boundary crossings of $w$-bit
words, hence pays fare $E \geq c_0
w\,\Omega(n^3/\sqrt S)$; a streaming blocked schedule achieves it.
\end{theorem}

\begin{proof}[Proof sketch]
Partition the execution into maximal segments of $\leq S$ crossings;
during a segment the fast level sees at most $2S$ distinct words of
each of the three classes, $A$-entries, $B$-entries, and $C$-contributions (partial sums evicted or written out count against the segment's crossings, as in the classical accounting), so by
the Loomis--Whitney argument of
\cite{HongKung81,IronyToledoTiskin04} it completes $\leq (2S)^{3/2}$
elementary products; there are $\geq n^3/(2S)^{3/2}$ segments. The
blocked upper bound promotes $\sqrt{S/3}\times\sqrt{S/3}$ tiles. The port to \emph{moving}-storage hierarchies (rings) is Appendix~\ref{app:hk}.
\end{proof}

\section{The Meeting Bound: Attention Pays $\Theta(T^2)$, State Pays $\Theta(T)$}
\label{sec:meeting}

When a language model writes the next word of an answer, it looks
back at every word of the conversation so far, through numbers it
stored for each of them (the key--value cache). That is exact causal
attention~\cite{Vaswani17}: with $T$ words, word $t$ must meet all
$t-1$ before it, and the meetings over a whole answer number about
$T^2/2$. Each meeting needs the old word's stored numbers to be
present, and the serving lemma prices presence. So the energy of a
long conversation grows with the square of its length, while a model
with a fixed-size memory grows linearly; this section states that as
theorems. Attention is the dominant workload of the decade, its memory
wall~\cite{WulfMcKee95} the one practitioners now hit, and the cleanest
modern instance of the law: the computation is
\emph{pairwise-complete}, like matrix multiplication's $n^3$ products.

\begin{definition}[$\kappa$-atomic meetings]\label{def:atomic}
An execution computes \emph{exact causal attention} over $T$ steps
$\kappa$-atomically if it keeps representations of its past tokens that are pairwise
disjoint as physical bits, of at least $\kappa$ bits each, and for
every pair $t < t'$, step $t'$
performs operations that together read every bit of token $t$'s representation at level 0, and recomputes a representation, when it does, only along routes dissipating at least $\gamma_{\mathrm{rc}}$ per bit (hypothesis (ii) of Lemma~\ref{lem:serve}; Table~\ref{tab:attention} prices the re-projection route).
Steps are time-disjoint (every operation of step $t'$ follows every
operation of step $t'{-}1$), and the \emph{inter-step time} $\Delta t$
is the least time from the last operation of a step to the first of
the next.
This hypothesis plays the role Hong--Kung's ``all $n^3$ elementary
products'' plays for matrix multiplication: it scopes the theorem to
the standard algorithm family. Quantization shrinks $\kappa$;
sparsity, windowing, and algebraic reformulation delete meetings and
exit exactness.
\end{definition}

\begin{remark}[The hypothesis is grounded, not circular]
\label{rem:kappa}
Two external results pin it. Space lower bounds by reduction from
communication complexity show that exact attention admits no
sublinear compression of its aggregate key--value state in the
$d = \Omega(\log n)$ regime: $\Omega(nd)$ bits are necessary and
$O(nd)$ numbers suffice for attention-based generation by one
attention function~\cite{HarisOnak25}, and the barrier extends to
tensor attention~\cite{KVlimits25}; the footprint is forced, not
chosen, up to the precision factor. And subquadratic-\emph{time} exact attention is
SETH-hard at large entries~\cite{AlmanSong23}: within the exact
regime, the meetings too reflect known algorithmic reality.
Disjointness, not just size, is what the charges need; Corollary~\ref{cor:aggregate} gives the form of the bound that needs only the \emph{joint} state's incompressibility.
\end{remark}

Before the conditional bound, the unconditional one: the space barrier
and the rent bridge together give a theorem about the \emph{function},
with no hypothesis on the algorithm.

\begin{theorem}[Unconditional meeting floor]\label{thm:uncond}
Call an execution of Definition~\ref{def:substrate} \emph{streaming}
if, for every $t$, the first crossing of a bit of token $t$ into a
working level comes after the last write of output $t{-}1$ into the
output region, as autoregressive decoding requires, and every bit of
a token crosses boundary $L$ exactly once (a token, once read, is
never re-read); this is a constraint on the execution, the
environment being the static tape of Definition~\ref{def:substrate}.
Let $\Delta t$ be the least time, over all streams and all steps,
between the first crossings of consecutive tokens. Then $\Delta t
\geq \tau_0$: output $t{-}1$ depends on token $t{-}1$ and an output
bit is what a fire wrote, so some fire reads a word of token $t{-}1$,
at least $\tau_0$ after that word crossed (Axiom~\ref{ax:lat}) and
hence after the token's first crossing, and token $t$'s first
crossing comes after that. (A machine allowed to read tokens ahead
could take every token in one round and make the interval below
empty; one token after the previous output is what makes the floor a
theorem. The words of one token may cross in different rounds; the
proof sends their value bits along with the message.) On any substrate whose working levels are
volatile, every streaming machine computing exact causal attention
over $T$ steps, for one attention function of key--value dimension
$d \geq C_u \log T$, pays, in expectation over the token
distribution $\mu$ constructed in the proof (keys fixed as
near-orthogonal projections of the unit vectors, values uniform
bits, queries zero), and hence on some stream,
\[
  E \;\geq\; p_{\min}\,\Delta t\,\Omega(d\,T^2).
\]
The mechanism is general: if under a distribution $\mu$ every
streaming machine has, for $t' = 2, \dots, T$, an expected minimum
of at least $s(t')$ live bits over the holds of the interval from
the first crossing of token $t'{-}1$ to the first crossing of token
$t'$, it pays
$E \geq p_{\min}\,\Delta t \sum_{t'=2}^{T} s(t')$; the proof
establishes that property for the $\mu$ above with $s(t') =
\Omega(t' d)$, through the Index reduction of Haris and
Onak~\cite{HarisOnak25}.
\end{theorem}

\begin{proof}
Fix a stream, a hold of the interval, and its last instant $u$. After
$u$ the machine reads from the input region only the words of token
$t'{-}1$ that have not yet crossed and the tokens $t', t'{+}1, \dots$,
so by Definition~\ref{def:substrate} every later output is a function
of the contents of the live cells at $u$, which by self-description
and (ii) determine which cells are live and which input word each raw
word holds, of those words, and of those tokens: two pasts with the
same contents at $u$ and the same uncrossed words receive the same
outputs from then on, and no timestamp is needed, since control reads
no clock and a conveyor's permutation is fixed. The protocol
of~\cite{HarisOnak25} makes this a one-way
message: Alice inserts tokens $1, \dots, t'{-}1$, runs the machine to
the end of a hold of the interval at which the live count, constant
through a hold, is least over its holds, and sends the computed
string, then the raw contents in the fixed cell order, then which
words of token $t'{-}1$ have not crossed by $u$ and the value bits of
those words, each prefix-free coded (by
self-description and (ii) the computed string determines the live
set and which input word each raw word holds, so the message
determines the state; the overhead is logarithmic in the length and
absorbed into the $\Omega$); Bob inserts a query token and reads one
value bit
off the exact attention output. The expected message length under
$\mu$ is therefore at least the distributional one-way complexity of
$\mathrm{Index}$ on $(t'{-}1)d$ uniform bits, which is $\Omega(t' d)$
(their exact-attention space bound, in the distributional form the
same information argument gives: an exact machine's message
determines every value bit); since the keys are fixed for the whole
stream and the values are uniform, one $\mu$ serves every $t'$. The
last part of the message is at most the $d$ value bits of token
$t'{-}1$ and a bitmap of its at most $3d\nu/w$ words, Bob filling the
keys and queries of the uncrossed words from $\mu$, under which they
are fixed; so the live count at $u$ is at least the message length
less $d + 3d\nu/w$ and the coding overhead, $\OO(d)$ in all. So the
minimum live count over the interval's holds is $\Omega(t' d) -
\OO(d) = \Omega(t' d)$ in expectation, for $t'$ above a constant that
the sum absorbs, the $s(t')$ of the general mechanism, throughout the
holds of an interval of length $\geq
\Delta t$, and the holds carry the interval's full duration, rounds
being instantaneous, each live bit at hold power $\geq p_{\min}$; the
intervals of distinct $t'$ are disjoint, and linearity of expectation
sums them.
\end{proof}

The proof asks of Bob only that he recover each value bit with
constant advantage, so the theorem holds as stated for a randomized
streaming machine whose outputs are $(1 \pm \eta)$-approximations of
the attention output with probability at least $9/10$ over its coins,
shared by Alice and Bob, which is the form in which Haris and Onak
state their bound (in bits, for approximate generation): under
uniform values and shared coins, a prefix-free message from which Bob
recovers each of $m$ independent bits with constant advantage has
expected length $\Omega(m)$, its entropy being at least its
information about the bits, $I(X; M \mid R) = \Omega(m)$, which is
the distributional form the expected minimum needs. We state the
exact form because exactness is what the meeting bound below prices.

With past tokens re-readable from the environment the floor becomes a
serving-lemma minimum of rent, re-read fare and recomputation, and
pricing the recomputation is exactly where the route hypothesis of
Definition~\ref{def:atomic} enters: the bound below is the sharpening.
It charges every past token's whole footprint $\kappa$, in fare, where
the unconditional floor charges $d$ bits of one attention function in
rent, and at Table~\ref{tab:attention}'s constants the two differ by
about seven orders of magnitude (four to five if the unconditional
floor is composed over the model's $640$ attention functions, one per
layer and key--value head, a composition the cited theorem does not
itself provide). The
unconditional floor is also asymptotic in $d$: its constant comes from
the Johnson--Lindenstrauss step of the reduction, which at $T = 2^{17}$
needs $d$ in the thousands (the standard bound $d \geq 8 \ln T /
\varepsilon^2$ at the reduction's $\varepsilon = 0.1$ gives $d
\approx 9 \times 10^3$), and its hard instance's softmax
exponentials, $e^{2\ln T/(1-2\varepsilon)} = T^{2/(1-2\varepsilon)}$,
exceed fp16's range (the query multiplier itself, $2\ln
T/(1-2\varepsilon)$, is small), so at $d = 128$
the theorem gives the form of the floor, not its value in the
arithmetic the systems use. (The theorem's ``exact'' is exactness
over the reals, a statement about the function; Bob's decoding needs
only that the softmax weight on the queried token exceed one half.)

\begin{theorem}[Meeting bound]\label{thm:meeting}
On a two-level substrate (level-0 capacity $S_0$, hold power $p_0$, fare
$c_0$, recomputation floor $\gamma_{\mathrm{rc}}$ for the routes of
Definition~\ref{def:atomic}), any $\kappa$-atomic execution of exact
causal attention with minimum inter-step time $\Delta t$ pays
\[
  E \;\geq\; \kappa \,\frac{T(T-1)}{2}\,\gamma, \qquad
  \gamma = \min\big(c_0,\; \gamma_{\mathrm{rc}},\; p_0\,\Delta t\big),
\]
with charges disjoint from the meetings' own operation energies.
($\Delta t$ is a parameter of the execution: a machine that overlaps
consecutive steps drives it toward zero, and
Corollary~\ref{cor:capacity} is the form that binds such a machine.)
\end{theorem}

\begin{proof}
1. Token $t$'s representation is a value of footprint $\geq \kappa$
created at step $t$ and read at steps $t{+}1, \dots, T$
(Definition~\ref{def:atomic}); charge, for each of its bits, the
first read at each step (a bit read several times within a step is
charged once there), so that consecutive charged uses lie in
consecutive steps, their gaps are $\geq \Delta t$, and the count is
still $\kappa$ per meeting. Lemma~\ref{lem:serve} charges every use;
over this sub-sequence the charges only fall, since $\min(a + b, c)
\leq \min(a, c) + \min(b, c)$ for the covered terms and a bit
uncovered over a sub-gap is uncovered over the union, so the lemma's
sum dominates the one below. 2. Definition~\ref{def:atomic} supplies hypotheses (i) and
(ii) of Lemma~\ref{lem:serve}. 3. The lemma charges each of the
$\sum_{t=1}^{T-1} (T - t) = T(T-1)/2$ (token, step) meetings at least
$\kappa \min(p_0 \Delta t, c_0, \gamma_{\mathrm{rc}})$, disjointly and
apart from the meetings' own operation energies.
\end{proof}

\begin{corollary}[Capacity form, timing-free]\label{cor:capacity}
Let $m = \lceil S_0/\kappa \rceil$. Then
\[
  E \;\geq\; \kappa\,\frac{(T-m)(T-m-1)}{2}\,
  \min\big(c_0, \gamma_{\mathrm{rc}}\big),
\]
independent of $\Delta t$: a floor no scheduling evades once the context exceeds the fast level.
\end{corollary}

\begin{proof}
Fix step $t'$. Every bit charged the covered amount at step $t'$
(Lemma~\ref{lem:serve}, sharp form) was resident at every instant
since its previous use at step $t'{-}1$; all those intervals contain
the gap between the last read of step $t'{-}1$ and the first read of
step $t'$ (the steps being time-disjoint), so all such bits are resident at level 0 in one common
state, the one after the last operation of step $t'{-}1$, and there are at most $S_0 \leq m\kappa$ of them (we count
every level-0 bit as key--value capacity, which only weakens the
bound). Every other bit met at step $t'$ is uncovered and charged
$\geq \min(c_0, \gamma_{\mathrm{rc}})$. Summing $\kappa \max(0,
t'-1-m)$ over $t' = 2, \dots, T$ gives the claim.
\end{proof}

\begin{corollary}[Aggregate form]\label{cor:aggregate}
Suppose instead that at every step $t'$ the execution reads at level 0
a joint representation $J_{t'}$ of tokens $1, \dots, t'{-}1$, completed
at step $t'{-}1$, of at least $K_{t'} \geq (t'-1)\,\kappa'$ bits, and
that hypothesis (ii) of Lemma~\ref{lem:serve} holds for the values
$J_{t'}$. Then $E \geq \kappa' \,\frac{T(T-1)}{2}\,\gamma$ with
$\gamma$ as in Theorem~\ref{thm:meeting}.
\end{corollary}

\begin{proof}
Each $J_{t'}$ is a value created at step $t'{-}1$ and read once, at
step $t'$, after a gap $\geq \Delta t$. $J_{t'+1}$ is completed
during step $t'$ while $J_{t'}$ is being read there, so the annotation
chooses the creation instant: it creates $J_{t'+1}$ at the later of
its completion and the last read of $J_{t'}$ within step $t'$,
relabeling the bits of $J_{t'}$ that survive, as
Lemma~\ref{lem:serve} allows, so the two are never simultaneously
live and hypothesis (i) holds; the relabeled bits start their terms
at zero, and the creation precedes step $t'{+}1$, whose reads follow
the inter-step gap, so the gap to the use is still $\geq \Delta t$.
Lemma~\ref{lem:serve} charges each $\geq K_{t'}\gamma$,
and $\sum_{t'=2}^{T} (t'-1)\kappa' = \kappa' T(T-1)/2$. This is the
form licensed by the aggregate space bound of~\cite{HarisOnak25}, with
$\kappa'$ the incompressible bits per token of the joint state.
\end{proof}

\begin{proposition}[State pays linearly]\label{prop:ssm}
A recurrent model holding a state of $\sigma$ bits at level 0 (state-space
models and their kin~\cite{GuDao23}) admits executions whose memory cost
is $T(p_0\,\sigma\,\Delta t + c_0\,\iota)$ for per-step I/O $\iota$, linear in $T$; conversely a model whose state is $\sigma$ incompressible bits live at working levels pays rent $\Omega(T\,p_{\min}\,\sigma\,\Delta t)$ on the same
accounting. Hence the attention/state separation on the memory terms
is $\Theta(\kappa T/\sigma)$. This compares the costs of two
functions, not of two algorithms for one: a $\sigma$-bit state
provably cannot compute what attention computes beyond $\sigma$
(copying, for instance~\cite{Jelassi24}; see
also~\cite{SanfordHT23}), so it is a separation of what is remembered,
priced, and not a separation in the complexity-theoretic sense.
\end{proposition}

\begin{table}[t]
\caption{The meeting bound with today's constants (illustrative, order-of-magnitude). The SRAM rent is for high-performance arrays (low-power arrays sit 1--2 orders lower); $\gamma_{\mathrm{rc}}$ prices the re-projection route at $\sim$1\,pJ/FLOP and is a hypothesis on the execution (Definition~\ref{def:atomic}), not a substrate floor; the state-model row and the unconditional row assume $\Delta t = 20$\,ms per token, the latter at DRAM self-refresh rent on $d = 128$ bits per token, an extrapolation of the cited theorem, whose constant is asymptotic in $d$ (Section~\ref{sec:meeting}), and, in the next row, composed over the model's $640$ attention functions (one per layer and key--value head), which is an assumption; remaining sources cited in text.}
\label{tab:attention}
\small
\begin{tabular}{@{}ll@{}}
\toprule
$\kappa$ (70B-class, 80 layers, GQA, fp16) & $\approx 320$\,KB/token $= 2.6$\,Mb \\
fare $c_0$, HBM-class~\cite{OConnor17,KecklerDally11} & $\approx 4$\,pJ/bit \\
rent $p_0$, SRAM-class hold & $\sim$1\,nW/bit \\
rent, DRAM self-refresh~\cite{LiuJVM12} & $\sim$0.1--1\,pW/bit \\
$\gamma_{\mathrm{rc}}$, KV re-projection & $\sim 10^3$\,pJ/bit \\
\midrule
fare floor at $T{=}128$k, whole generation & $\kappa T^2 c_0/2 \approx 9{\times}10^4$\,J \\
unconditional floor, streaming (Theorem~\ref{thm:uncond}): & \\
\quad one attention function (extrapolated) & $\approx 2$--$20$\,mJ \\
\quad composed over $640$ of them (assumed) & $\approx 1.4$--$14$\,J \\
fare at the last token & $\kappa T c_0 \approx 1.4$\,J \\
model arithmetic per token (constant in $T$) & $\approx 0.1$--$0.3$\,J \\
crossover $T^*$ (fare $=$ arithmetic) & $\approx 1$--$3 \times 10^4$ \\
all-SRAM rent alternative ($\kappa T = 42$\,GB) & $\approx 340$\,W standing \\
state model, $\sigma = 10^8$--$10^9$\,b memory term & $\approx 0.002$--$0.02$\,J/token \\
\bottomrule
\end{tabular}
\end{table}

\noindent
\textbf{Reading the numbers} (Table~\ref{tab:attention}, drawn as
Figure~\ref{fig:frontier}; fares and ops
from \cite{KecklerDally11,Horowitz2014}, refresh from \cite{LiuJVM12}).
Under the route hypothesis of Definition~\ref{def:atomic} (recomputation by re-projection), today $\gamma = c_0$: the binding currency is fare, and the bound says
energy per generated token grows \emph{linearly in context length} with
slope $\kappa c_0$ while the model's arithmetic stays constant, so the KV term overtakes near $T^* \approx$ tens of thousands of tokens, which
is precisely where long-context serving becomes bandwidth-bound in
practice. The known systems are placements on the serving lemma's
frontier: FlashAttention streams the meetings (fare-payer, and
fare-optimal: Saha and Ye prove its I/O count optimal for the batched
computation by red--blue pebbling~\cite{SahaYe24}, the fare leg of the
bound for prefill, where ours prices decode in three
currencies)~\cite{Dao22}; paged KV manages where the rent is
paid~\cite{KwonSOSP23}; full-SRAM residency is the rent-payer (priced:
infeasible capacity, hundreds of watts standing); recomputing KV each
step is the ops-payer ($\gamma_{\mathrm{rc}} \sim 10^3\pJ$/bit; it loses by orders); gradient checkpointing is the same trade run during
training, action exchanged for
recomputation~\cite{ChenXZG16,GriewankWalther00}. Batching $B$
sequences amortizes the weights' fare but not the key--value fare,
which is per sequence: the bound predicts that decode becomes
bandwidth-bound in the key--value term at large batch, which is what
serving systems observe. Recurrent state
\emph{exits} the meeting count itself (Proposition~\ref{prop:ssm}):
the attention-versus-state debate is, in our account, an energy theorem, with the seed-residency floor of Theorem~\ref{thm:action} as its miniature in the model of our instrument. The bound also tells the auditor what a sub-floor measurement
would \emph{mean}: the system is not $\kappa$-atomic, so its effective footprint is smaller (compression) or its meetings fewer (approximation).

\begin{figure}[t]
\centering
\begin{tikzpicture}
\begin{axis}[
  width=0.96\columnwidth, height=8.4cm,
  xmode=log, ymode=log,
  xmin=1e3, xmax=1e6, ymin=3e-10, ymax=3e3,
  xlabel={context length $T$ (tokens)},
  ylabel={joules per generated token},
  xtick={1e3,1e4,1e5,1e6},
  ytick={1e-9,1e-7,1e-5,1e-3,1e-1,1e1,1e3},
  tick label style={font=\scriptsize},
  label style={font=\footnotesize},
]
\fill[black!8]  (axis cs:9.5e3,3e-10) rectangle (axis cs:2.9e4,3e3);
\fill[black!15] (axis cs:1e3,0.1)    rectangle (axis cs:1e6,0.3);
\fill[black!15] (axis cs:1e3,0.002)  rectangle (axis cs:1e6,0.02);
\addplot[domain=1e3:1e6, samples=2, line width=0.9pt] {1.048e-5*x};
\addplot[domain=1e3:1e6, samples=2, dashed, line width=0.7pt] {2.62e-3*x};
\addplot[domain=1e3:1e6, samples=2, densely dotted, line width=0.9pt] {5.24e-5*x};
\addplot[domain=1e3:1e6, samples=2, dashdotted, black!60, line width=0.6pt] {2.56e-12*x};
\addplot[only marks, mark=*, mark size=1.8pt] coordinates {(1.3e5,1.37)};
\addplot[only marks, mark=square*, mark size=1.7pt] coordinates {(1.3e5,6.8)};
\addplot[only marks, mark=triangle*, mark size=2.2pt] coordinates {(1.3e5,341)};
\node[font=\scriptsize, anchor=west] at (axis cs:1.15e3,12) {recompute};
\node[font=\scriptsize, anchor=west] at (axis cs:1.15e3,3.1e-2) {rent per step};
\node[font=\scriptsize, anchor=west] at (axis cs:7e4,0.36) {fare};
\node[font=\scriptsize, anchor=north west] at (axis cs:1.15e3,2.4e-9) {unconditional floor};
\node[font=\scriptsize, anchor=west] at (axis cs:3e4,0.16) {model arithmetic};
\node[font=\scriptsize, anchor=west] at (axis cs:1.15e3,6.4e-3) {recurrent-state memory term};
\node[font=\scriptsize, anchor=south] at (axis cs:1.66e4,4e-10) {$T^*$};
\node[font=\scriptsize, anchor=north west, align=left] at (axis cs:1.5e5,0.95) {FlashAttention,\\paged KV};
\node[font=\scriptsize, anchor=east] at (axis cs:1.12e5,9.0) {full-SRAM residency};
\node[font=\scriptsize, anchor=east] at (axis cs:1.15e5,470) {KV recomputed each step};
\end{axis}
\end{tikzpicture}
\caption{The serving frontier of exact attention at the constants of
Table~\ref{tab:attention} (illustrative, order of magnitude). Every
past token's $\kappa$ bits must be served at every step, by fare
($\kappa c_0 T$ per generated token), by rent since the previous step
($\kappa p_0 \Delta t\,T$), or by recomputation
($\kappa\gamma_{\mathrm{rc}}T$), so each pure strategy is a line of
slope one and a platform pays the lowest, today fare
(Theorem~\ref{thm:meeting}); the marked systems are placements on the
frontier, not measurements. The fare line passes the model's own
arithmetic in the band $T^* \approx 1$--$3\times10^4$, where
long-context serving becomes bandwidth-bound in practice; a recurrent
state exits the meeting count and stays flat
(Proposition~\ref{prop:ssm}); the unconditional streaming floor
(Theorem~\ref{thm:uncond}; one attention function at DRAM-refresh
rent, drawn at the upper end of Table~\ref{tab:attention}'s band,
$1$\,pW per bit, an extrapolation) shows what is proven with no serving
hypothesis, six to seven orders below the fare leg.}
\Description{Log-log plot of joules per generated token against
context length, with three slope-one serving lines for fare, rent per
step, and recomputation, shaded bands for model arithmetic and
recurrent-state memory, the crossover strip near ten to thirty
thousand tokens, the unconditional streaming floor far below, and
three marked systems placed on the frontier.}
\label{fig:frontier}
\end{figure}
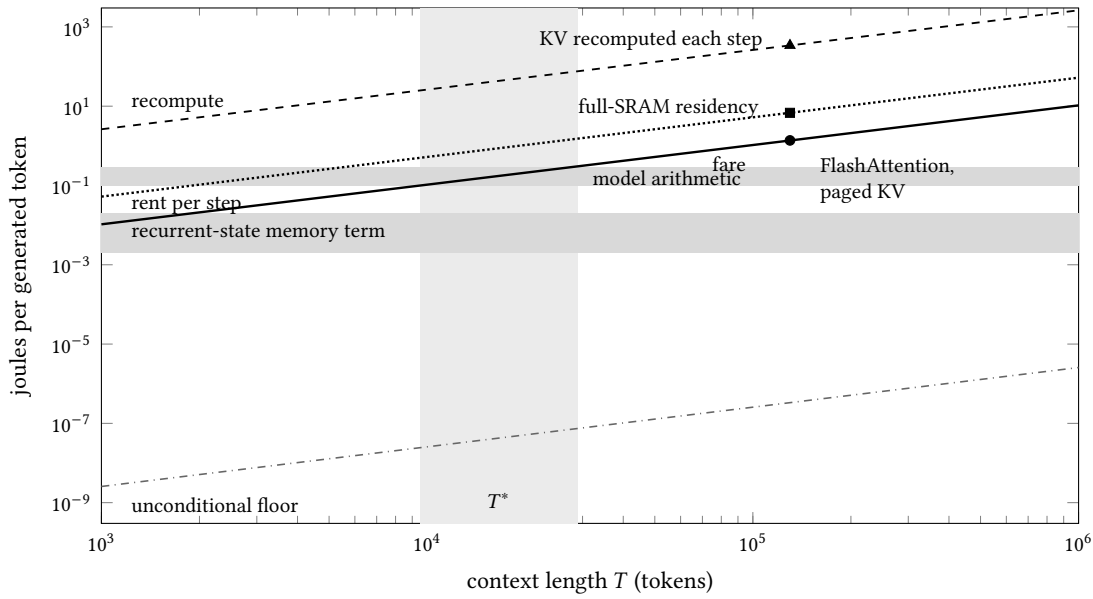

\section{Structure: Why Hierarchies Exist}
\label{sec:structure}

The bridges and the meeting bound price \emph{how much} keeping and
moving cost. The structural half asks \emph{why a hierarchy of
levels exists at all}. Necessity is the asymptotic answer and arrives
late; the practical answer is that a hierarchy converts rotation into
parking: what is not needed now sits still on a cheap shelf instead
of riding the conveyor, and a machine with only one conveyor has no
shelf. In the model that abstracts our calibrating machine, storage
that circulates, where layout (``seating'') is exact and every
construction the compiler realizes is machine-checked, we prove the
separation in rotation action (Theorem~\ref{thm:action}), price
parking on the instrument at about a fiftieth of rotation per
slot-cycle, and leave open whether the trade beats the best single
ring in energy (Open Problem~9); the
seating theory around the main theorem is stated below in brief and
developed, statements and proofs, in Appendix~\ref{app:proofs}.

\begin{definition}[Cyclic-storage machine]\label{def:csm}
A \csm{} stores packets (values, instructions, relays, or empty) in
rings (cyclic arrays of slots advancing one position per cycle) and computes at fixed stations: an instruction packet arriving at a
station reads operands $\leq W$ slots ahead of itself and rewrites a
packet $\leq W$ ahead; a \emph{relay} copies a packet a bounded
distance; a \emph{transfer} moves a packet between rings at a scheduled
phase; a ring may be \emph{parked} (its slots hold without advancing,
and its stations fire nothing: a station fires on an instruction
packet arriving in front of it, and on a parked ring nothing
arrives).
A ring is a storage level in the sense of
Definition~\ref{def:substrate}, its stations at that level, whose
conveyor is its rotation, at
hold power the shift constant $e$ while it advances, and a hierarchy of rings is a substrate with one level per ring and one boundary, at the landing fare charged per crossing in either direction, between each pair of rings joined by a transfer station, a graph of boundaries as Definition~\ref{def:substrate} allows, a transfer crossing one of them (realized as moves and two occupancy fires, Section~\ref{sec:instrument}); parked, it is
the same level with its conveyor at rest, at hold power
$p_{\mathrm{park}} > 0$ by Axiom~\ref{ax:vol} (Definition~\ref{def:substrate} gives a level one hold power running and one at rest), on our instrument the clock-gated slot of Table~\ref{tab:events}, switching plus leakage. A \emph{seating} is an injective assignment of
a program's packets to slots. The \emph{rotation action} of a run is
$\act = \sum_t \ell_t$, live packets summed over advancing cycles (the model form of the law's rent, and, because the storage moves, of its fare too), and its \emph{parked action} $\act_{\mathrm{park}}$
is the same sum over parked cycles, priced at $p_{\mathrm{park}}$
rather than at the shift constant $e$; we report both in every action comparison below. Settle semantics re-fires idempotently every revolution,
so a run's length is fixed by its \emph{output convention}, a designated (ring, slot, cycle), and every action comparison we make is between minimal correct runs, which are finite executions in
the sense of Definition~\ref{def:substrate}.
\end{definition}

Four facts about single-ring seating, stated here in brief and
developed, statements, definitions, and proofs, in
Appendix~\ref{app:proofs}, set the stage. \emph{Necessity is
asymptotic}: a correct single-ring seating with window $W$ forces
routed cyclic value-cutwidth $\leq 2W$ and linear value-cutwidth
$\leq 4W$ (Theorem~\ref{thm:necessity}), so graphs of larger cutwidth
cannot seat at all. \emph{Sufficiency nearly meets it}: strand load
$\lambda \leq \lfloor (W-5)/2 \rfloor$ seats on one ring by an
explicit lattice, one strand short of the counting ceiling
(Theorem~\ref{thm:lattice}), and at $W = 8$ the lattice's capacity is
exactly the discipline our compiler realizes. \emph{Optimal layout is
intractable}: the minimal window is NP-hard to approximate within any
constant factor, by a sandwich $\mathrm{bw} + 1 \leq W^* \leq
5\,\mathrm{bw} + 1$ against caterpillar bandwidth
(Theorem~\ref{thm:nphard}), while seatability at fixed $W$ is
decidable in $n^{\OO(W)}$ time (Remark~\ref{rem:saxe}, a sketch).
And \emph{systolic arrays embed}:
every unidirectional moving-results design seats on a chain of rings
at constant window and constant slowdown
(Proposition~\ref{prop:systolic}), our streaming tap and two-stage
pipeline being its one- and two-cell instances, while
stationary-result designs map to period gearing, the accumulator
recurring at the ring period as the stream passes at the beat period,
which is exactly our measured dot product and MAC chain. The \csm{}
is, in this sense, the stored-program closure of the systolic array:
the rhythm survives; the function becomes data.

\begin{corollary}[Wide reductions force a hierarchy, asymptotically]\label{cor:trees}
The balanced $N$-leaf reduction has value-cutwidth $\cw =
\Theta(\log N)$: every value has one consumer, so it is the edge
cutwidth of the complete binary tree, which is
$\Theta(\log N)$~\cite{Lengauer82} and dominates vertex separation
$=$ pathwidth $\geq \tfrac12 \log_2 N$~\cite{EllisSudboroughTurner94}),
so by Theorem~\ref{thm:necessity} (against $\cw \leq 4W$) it has
\emph{no} single-ring seating once $N > 2^{8W}$, while a staged
hierarchy seats it (Theorem~\ref{thm:action}). We state the threshold
with its constant because the constant decides what the corollary
says: at our instrument's $W = 8$ it is $N > 2^{64}$. Below it, which
is every reduction anyone will run, single-ring seatings exist and the
hierarchy is a matter of rotation action (Theorem~\ref{thm:action}),
whose energy form is Open Problem~9; beyond it the hierarchy is
necessary. Necessity is the asymptotic fact; the separation in
rotation is the practical one.
\end{corollary}

\begin{theorem}[Action separation]\label{thm:action}
Let the \emph{$N$-leaf reduction} be the balanced binary tree of
$N{-}1$ binary operations over $N$ seed values, $N/2$ of which consume
two seeds each. Any single-ring program computing it with $s$ stations at which a seed-consuming operation can fire (one, the multiplier, for the products of our instrument's kernels), where inputs enter only as seeded packets and kills after last use are allowed, has rotation action $\act \geq N^2/(4s)$. A staged
hierarchy on $\OO(N/b)$ rings of period $\OO(b)$, one $b$-leaf subtree
per ring with roots carried upward by scheduled transfers, for the
stage size $b = 2^{\lambda+1}$ that the window's strand capacity
$\lambda = \lfloor (W-5)/2 \rfloor$ permits ($b = 4$ at $W = 8$),
achieves rotation action $\act = \OO(N b)$, parked action
$\act_{\mathrm{park}} = \OO(N^2)$, and $\OO(N/b)$ transfers. Hence the
rotation-action ratio is $\Omega(N/(s\,b))$, linear in $N$ for fixed
$W$. Priced, with parked rent at the resting hold power of
Definition~\ref{def:csm}, the hierarchy's bill is asymptotically its
parked rent, $p_{\mathrm{park}}\,c_2 N^2$ for a constant $c_2$ of the
construction, against the single ring's floor $e\,N^2/(4s)$ in
rotation; the ratio of that floor to the hierarchy's bill tends to
$e/(4s\,p_{\mathrm{park}}\,c_2)$, so the theorem separates energy
only where $c_2 < e/(4s\,p_{\mathrm{park}})$, which is about $13$ at
$s = 1$ on our instrument ($e/p_{\mathrm{park}} \approx 54$,
Table~\ref{tab:events}, leakage included). Table~\ref{tab:parked}
reports $c_2$ on the instrument; it is above that value.
\end{theorem}

\begin{proof}[Proof of the lower bound]
All $N$ seeds are live from cycle 0, and each lives at least until the
operation consuming it fires; a station fires at most once per cycle,
so the $r$-th seed-consuming operation fires no earlier than
advancing cycle $\lceil r/s \rceil$, $r = 1, \dots, N/2$ (a parked
ring accrues no rotation action and fires nothing,
Definition~\ref{def:csm}). Then $\act \geq 2
\sum_{r=1}^{N/2} \lceil r/s \rceil \geq \tfrac{N}{2}(\tfrac{N}{2}+1)/s
\geq N^2/(4s)$. This is Lemma~\ref{lem:serve} with rent the only currency; it is the meeting bound of Section~\ref{sec:meeting} in miniature, which is why the two share a shape. The upper bound is
Appendix~\ref{app:action}: one $b$-leaf subtree per ring, seated by
Theorem~\ref{thm:lattice}, roots carried upward by scheduled
transfers, subtrees run one at a time on parked rings; the $\OO(N)$
seeded packets wait parked for the $\OO(N)$-cycle run, which is the
$\OO(N^2)$ parked action. Its energy $p_{\mathrm{park}}
\act_{\mathrm{park}}$ is the hierarchy's asymptotic bill.
\end{proof}

The separation in rotation is parking, not parallelism: the stages
run one at a time on one active station, and the ratio comes entirely
from the $\OO(N)$ seeds that stand still instead of rotating. The
comparison is at the single ring's $s$ stations: the hierarchy
provisions $\Theta(N/b)$ stations and fires one at a time, and at $s
= \Theta(N/b)$ the bound is $\Omega(1)$, so we claim the separation
at fixed $s$ only. Provisioned stations leak: in our artifact an idle
ALU station leaks about $0.14\pJ$ per cycle and an idle MUL station
with its ring's two landing ports and write network at most $0.31$
(the compute ring's leakage less its slots and the ALU), and the
hierarchy provisions stations and ports on every ring, $0.45\pJ$ per cycle against
$0.28\pJ$ for the sixteen parked slots of a $P{=}16$ ring; charged,
that more than doubles the parked term per ring and moves every
ratio below in the same direction. The single ring's floor is
the residency argument of Remark~\ref{rem:input} with the input on the
substrate, which on a ring it necessarily is: a ring cannot use the
environment's free hold, which is exactly why it needs a parked level.

Feasibility is not the issue; economy is. A reduction \emph{does}
flatten onto one sufficiently long ring: seat the stages wall-free and connect them by static relay ladders (respecting Lemma~\ref{lem:wall}). We machine-checked the eight-seed instance (four products and three sums, on rings that carry an ALU and a MUL station; every seed-consuming operation is a product, so Theorem~\ref{thm:action}'s $s$ is one): 13 ladder hops, correct to the reference. At eight seeds the floor of
Theorem~\ref{thm:action} is trivial (sixteen packet-cycles; the sum
itself gives twenty); what the
flattened ring pays is the ladders' price: the 13 hops settle serially,
each a revolution or more behind its source, so the run lasts 25
revolutions of a $P{=}128$ ring while about 47 packets stay live, for $151{,}031$ packet-cycles. The staged program (two $P{=}16$ rings, 129
cycles; our compiler cuts at subtree boundaries and made two stages
where Theorem~\ref{thm:action}'s $b = 4$ allows one, which would only
improve the ratio) pays $1{,}820$ rotating and $1{,}792$ parked
packet-cycles:
$\mathbf{83\times}$ in rotation action, $82\times$ in energy at our
instrument's rates ($e \approx 0.95\pJ$, $p_{\mathrm{park}} =
0.0177\pJ$ per slot-cycle with leakage; charging leakage to every
provisioned slot of both programs moves the ratio by less than one
unit), the parked rent being $2\%$ of the hierarchy's bill. The hierarchy is
the same program with the ladders deleted.

\textbf{What the $83\times$ shows, and what it does not.} Both
programs sit far above the theorem's floor: the flattening at nine
thousand times it in rotation, the staged program at over a hundred
times it in energy. The theorem is silent about the ratio of two
programs both far above its floor, and nothing in this paper bounds
the best single-ring program between sixteen and $151{,}031$
packet-cycles. What the $83\times$ measures is
the price of one flattening's thirteen relay ladders, twenty-five
revolutions of a $P{=}128$ ring, against a transfer: a measurement
of one flattening, not an instance of the theorem.

\textbf{The parked constant.} Table~\ref{tab:parked} runs the
construction of Theorem~\ref{thm:action} as our compiler realizes it,
at eight to sixty-four seeds, priced at the instrument's rates. The
rotating action grows linearly, about $500$ packet-cycles per seed,
as the theorem says. The parked action grows as $c_2 N^2$ with $c_2$
about $10^2$ packet-cycles per seed squared and still rising at $N =
64$: a run of about $28$ cycles per seed, with about four packets per
seed parked through it. The single ring's floor $e N^2/4$ (one
multiplier for all the products, $s = 1$) is below the hierarchy's
bill in every row and in the limit, where the ratio of the floor to
the bill tends to $e/(4\,p_{\mathrm{park}}\,c_2) \approx 0.13$;
station leakage, charged per provisioned ring, moves it the same
way. So the theorem with its
constant filled in does not show the hierarchy cheaper than the
single ring's floor at any $N$. What it shows is that the hierarchy's
asymptotic cost is parked rent while the single ring's is rotation,
and where the balance point sits: on the compiler's constants the
hierarchy's rotation alone drops below the floor only past $N \approx
2{,}000$ seeds ($500 N = N^2/4$). A separation in energy, theorem against
theorem, would need a single-ring lower bound that charges what a
flattening actually pays, the relay traffic that cutwidth
$\Theta(\log N)$ forces at window $W$ and the revolutions a ladder
costs; that bound is not in our toolkit (Open Problem~9).

\begin{table}[t]
\caption{The staged construction of Theorem~\ref{thm:action} from our
compiler at $N$ seeds on rings with an ALU and a MUL station, on the
cycle-accurate model (run in cycles; rotating, parked and floor in
packet-cycles), energy at the instrument's rates ($e = 0.95$,
$p_{\mathrm{park}} = 0.0177\pJ$ per slot-cycle). The floor is the
single ring's $N^2/4$: every seed-consuming operation is a product
and the ring has one multiplier, Theorem~\ref{thm:action}'s $s = 1$.
The eight-seed row is the
staged program of the $83\times$ comparison; the flattening it beat
pays $151{,}031$ rotating packet-cycles. The table regenerates from
the artifact's certificate T5.}
\label{tab:parked}
\small
\begin{tabular}{@{}rrrrrrrrr@{}}
\toprule
$N$ & rings & run & rotating & parked & parked$/N^2$ & floor & hierarchy & floor \\
 & & cycles & pkt-cyc & pkt-cyc & & pkt-cyc & pJ & pJ \\
\midrule
8  & 2  & 129   & 1{,}820  & 1{,}792   & 28  & 16    & 1{,}761  & 15 \\
16 & 4  & 451   & 7{,}373  & 19{,}546  & 76  & 64    & 7{,}350  & 61 \\
32 & 8  & 903   & 15{,}494 & 94{,}531  & 92  & 256   & 16{,}392 & 243 \\
64 & 16 & 1{,}807 & 31{,}793 & 413{,}350 & 101 & 1{,}024 & 37{,}520 & 973 \\
\bottomrule
\end{tabular}
\end{table}

\section{The Dependence Tax}
\label{sec:boundary}

Static scheduling has a boundary, and the law prices it. A program
whose every access is known in advance can be laid out so that each
value arrives exactly when needed. A pointer chase cannot: the next
address is known only after the current read, and on a rotating store
the requested entry is, on average, half a revolution away.

\begin{theorem}[Half-orbit tax]\label{thm:halforbit}
Hold an $n$-entry table $R$ at fixed slots of a ring of period $P$,
each entry in exactly one slot, read through one port (a station), and run an
$L$-step chase $j_{k+1} = f(R[j_k])$; let $j_1$ be uniform on the
table and each $j_{k+1}$ uniform on the table given $j_1, \dots, j_k$
(as it is, up to the first repeat, for a table of independent uniform
entries and a bijection $f$), the expectation being over the
sequence. For
any placement and any mechanism,
$\;\mathbb{E}[T] \geq (L-1)\,\frac{P}{2}\,(1 - \frac1n)$.
\end{theorem}

\begin{proof}
The station meets entry $j_k$ only at cycles $\equiv \phi_k \pmod P$,
and reads are dependence-ordered, so $T \geq \sum_k ((\phi_{k+1} -
\phi_k) \bmod P)$. Under the hypothesis $(j_k, j_{k+1})$ is uniform on
$[n]^2$. For distinct $a, b$: $((b-a) \bmod P) + ((a-b)
\bmod P) = P$; pairing each ordered pair of table phases with its
reverse gives $\mathbb{E}[(\phi_{k+1}-\phi_k) \bmod P] =
\frac{P}{2}(1 - \frac1n)$ exactly.
\end{proof}

Placement cannot help: after a service the port sits \emph{inside} any
cluster, half the next targets behind it. The 1970s drum literature
optimized rotational latency by record placement~\cite{CodyCoffman76};
this is its converse: for dependent uniform chases, the half orbit is a floor. Our instrument meets it within one cycle, with one held request,
one phase counter, one comparator: wait $= ((\mathit{POS} - j - t)
\bmod P)$, verified at all 64 phases of a $P{=}64$ ring, mean $32.5 =
(P{+}1)/2$ against the bound's $31.5$ at $n = P = 64$; the difference is the same-index case, where the bound charges nothing and
the mechanism waits a full orbit. The phase counter is the one place
the instrument holds its phase as a number: a wait computed from a
run-time index needs it, so it is a live register and pays rent,
whereas the ring itself needs none, its phase being the layout of
its state (Section~\ref{sec:instrument}). Replication is the escape the
hypothesis excludes: $k$ copies of the table divide the floor by $k$
at $k$-fold rent, the serving lemma again, and so do $k$ ports at
$k$-fold provisioned cost. Against a unit-cost RAM the tax is
$\Theta(P)$ per dependent access; against a real hierarchy at equal
capacity the ring's $P/2$ is worse than an SRAM and better than a DRAM
row miss. Pointer-dense workloads whose working set fits fast storage
belong to random-access machines, and the law says by how much.

\section{The Instrument: Where the Law Is an Equality}
\label{sec:instrument}

Lower bounds need constants; constants need a machine where the law is
tight. Our calibrating instrument is MADAR~\cite{Bergach2026}: a
processor with no program counter, register file, or cache, whose
storage is a hierarchy of rings in which instructions and data
co-rotate and stations fire on scheduled coincidence. In plain terms
it is a conveyor belt: packets of data and instructions ride around a
ring, a station fires when an instruction passes it, and a packet
costs energy for every slot it advances, whether one calls that
keeping it or moving it. Because the storage \emph{moves}, rent and
fare collapse into one constant (a live
packet pays its shift each cycle whether viewed as holding or as
transport; within a ring the accounting of \eqref{eq:law} has $c = 0$
and $p = e$, one event charged once, while the boundary between rings
is a landing, priced as fare in Table~\ref{tab:identity}), which makes
the machine our natural measuring instrument for the action
functional. In the terms of Definition~\ref{def:substrate} a ring
is a level whose conveyor is the rotation; a slot's payload field is
its word, and the slot's header and occupancy cells are cells outside
every word, computed always, so that a seeded data packet is a raw
word under a computed header and every slot's occupancy flop is a
live cell whether the slot is occupied or not. A landing from another
ring is, within one round, the moves of the payload bits into the
free cells of the destination word (the landing-port fare) followed
by two fires that read nothing, one at each ring's level, writing the
source slot's occupancy cell to zero and freeing its payload cells,
and the destination's to one,
both control decisions the string determines, which is how the string
records the move as clause (ii) requires; a kill is a fire that
overwrites the raw word and frees its cells, writing the occupancy
cell to zero. The station reads the cell in front of it, and the
phase, which slot is in front, is not a register but the layout of
the state. The idle floor that Table~\ref{tab:identity} subtracts,
$0.0074\pJ$ per slot-cycle on every empty slot, is in these terms
the rent of those occupancy cells: what the ring pays to know where
its packets are (an occupied slot's occupancy cell is among the $81$
flops the clock term of Table~\ref{tab:fit} prices). What it calibrates is the \emph{form}: that the residual
after operations, fares and idle floors is proportional to action
with a constant that does not depend on the program beyond the
activity of its data, an activity the constant itself resolves
(Table~\ref{tab:fit}). Its constant is not the law's constant
anywhere else: at $10$\,fJ per bit per cycle ($0.82\pJ$ over an
$81$-bit slot), $10$\,$\mu$W per bit at $1$\,GHz, the ring's rent is
four orders above the leakage rent of an
SRAM cell (Table~\ref{tab:attention}); every number in
Table~\ref{tab:attention} is a datasheet number; and
Axiom~\ref{ax:vol}'s content for static storage, leakage and refresh,
is tested by Prediction~2, not by the ring.

Our artifact synthesizes the rings and stations (and, through the same
flow, a baseline RV64 in-order core) to the open NanGate 45\,nm
library, proves the gate netlist equivalent to the abstract model by
running the same programs (slot-for-slot output equality across the
suite; the Python model, the RTL, and the netlist agree), and extracts
per-event energies from switching activity by the difference method
(Table~\ref{tab:events}).

\begin{table}[t]
\caption{Per-event energies extracted at gate level (NanGate45, 1.1\,V,
25\,$^\circ$C, typical; switching energy per event, which is
frequency-independent, and leakage where stated at the flow's 1\,GHz
clock; no wire capacitance or clock tree; method and caveats in the
artifact).}
\label{tab:events}
\small
\begin{tabular}{@{}lr@{}}
\toprule
event & pJ \\
\midrule
occupied slot advance, 64-bit random payload & 0.623 \\
occupied slot advance, 8-bit payload & 0.069 \\
idle slot, ungated (clock net + flop pins) & 0.677 \\
idle slot, clock-gated & 0.0074 \\
ring slot leakage per cycle at 1\,GHz, occupied or not (stations excluded) & 0.0103 \\
same, on the ungated ring (hold-mux slot) & 0.0145 \\
parked slot, switching $+$ leakage ($p_{\mathrm{park}}$) & 0.0177 \\
occupied-slot marginal on the gated ring (incl.\ ICG) & 1.268 \\
ADD firing (vs.\ bubble; incl.\ operand muxes) & 11.3 \\
$64{\times}64$ MUL firing & 23.0 \\
\quad same, 8-bit operands & 5.7 \\
station churn on passing data, all stations (worst case) & 27.0 \\
\quad same, operand-isolated & 4.1 \\
\quad of which the MUL station & 25.3 \\
ALU station, switching per cycle on passing data, no firing & 4.5 \\
landing-port write (inter-ring) & 14.2 \\
\midrule
parked sum loop, per instr: ungated / gated & 43.8 / 19.5 \\
RV64 in-order baseline, same flow, per instr & 28.4 \\
streaming dot, per MAC: base / gated+isolated & 59.7 / 49.8 \\
streaming MAC, random data, per MAC: base / gated+isol.\ & 224.9 / 119.4 \\
systolic cell (synthesized), per MAC: 8-bit / random 32- and 64-bit & 3.0 / 16.0 / 23.0 \\
\quad same, inputs held (clock floor of 192 flops), per cycle & 1.64 \\
\bottomrule
\end{tabular}
\end{table}

\begin{table}[t]
\caption{The identity across programs (gate level, gated hardware;
pJ). Each program has a no-fire twin (every instruction packet turned
into a data packet with the same bits: nothing fires, the same packets
rotate, landings kept) and a static twin (no landings). Ops $=$
program $-$ no-fire; fare $=$ no-fire $-$ static; rent $=$ static $-$
idle floor, the floor being the rent of the empty slots' occupancy
cells (Section~\ref{sec:instrument}); $e = $ rent$/\act$ with $\act$ the
model-computed action.
In parentheses, what the random-data per-event constants of
Table~\ref{tab:events} would have charged instead. The last two rows
are static runs, eight random 64-bit packets rotating on the compute
ring with no operations, with and without operand isolation.}
\label{tab:identity}
\footnotesize
\setlength{\tabcolsep}{3pt}
\begin{tabular}{@{}lrrrrrr@{}}
\toprule
program & $\act$ & $E$ & ops (fixed) & fare (fixed) & rent & $e$ \\
\midrule
sum, 20 ADD & 800 & 838.5 & 2.3 (225) & 0 & 823.1 & \textbf{1.029} \\
dot, 4 MAC & 192 & 199.3 & 15.0 (137) & 10.7 (114) & 173.2 & \textbf{0.902} \\
MAC, random, 20 & 960 & 2388.5 & 1148.6 (686) & 359.9 (568) & 877.6 & \textbf{0.914} \\
rotation, random, isolated & 1280 & 2311.0 & 0 & 0 & 2311.0 & \textbf{1.805} \\
rotation, random, not isol. & 1280 & 5984.3 & 0 & 0 & 5984.3 & \textbf{4.675} \\
\bottomrule
\end{tabular}
\end{table}

\begin{table}[t]
\caption{The rent constant resolved: $e = e_{\mathrm{clk}} +
e_{\mathrm{tog}}\,\tau$, with $\tau$ the net transitions per live
slot-cycle counted over the whole netlist in the static run. Least
squares over the five static runs on bare or operand-isolated
hardware (the ALU ring is not isolated; its data are small; the
bare ring is the sixteen-slot gated ring without stations, the
artifact's \texttt{r16g\_full}):
$e_{\mathrm{clk}} = 0.822\pJ$ per live slot-cycle, $e_{\mathrm{tog}} =
3.4$\,fJ per transition. The un-isolated run is held out.}
\label{tab:fit}
\small
\begin{tabular}{@{}lrrrr@{}}
\toprule
run & $\tau$ & $e$ & fit & residual \\
\midrule
dot, static & 30.0 & 0.902 & 0.924 & $-2.3\%$ \\
MAC, static & 30.0 & 0.914 & 0.924 & $-1.0\%$ \\
sum, no-fire & 50.4 & 1.029 & 0.993 & $+3.7\%$ \\
bare ring, random data & 133.0 & 1.268 & 1.272 & $-0.3\%$ \\
compute ring, random, isolated & 290.9 & 1.805 & 1.807 & $-0.1\%$ \\
same, not isolated (held out) & 904.2 & 4.675 & 3.883 & $+20.4\%$ \\
\bottomrule
\end{tabular}
\end{table}

\textbf{The identity.} The law's terms must be measured on the
program's own data, and the difference method does that
(Table~\ref{tab:identity}): each gate-measured program on the gated
hardware has a \emph{no-fire twin}, the same packets with every
instruction turned into a data packet of identical bits, so that
nothing fires and the packets rotate as before, and a \emph{static
twin} without landings. Program minus twin is the operations' energy;
twin minus static is the landings' fare; static minus the idle floor
is rent; and rent divided by the model-computed action $\act$ is the
implied constant. Across the three programs, whose operation energies
differ by a factor of $100$ (a counted sum of small integers, a
four-term dot product, a twenty-term multiply--accumulate on random
32-bit operands), $e$ is $0.90$ to $1.03\pJ$ per live-slot-cycle. That
is not the whole story, and our artifact says so: random 64-bit
packets rotating with no operations give $1.8\pJ$ on the isolated
compute ring and $4.7\pJ$ on the un-isolated one (the last two rows).
The three programs agree because their rotating data are alike, $30$
to $50$ net transitions per live slot-cycle; a single constant does
not survive data of higher activity. Two constants do
(Table~\ref{tab:fit}): with $\tau$ the net transitions per live
slot-cycle,
\[
  e \;=\; 0.82\pJ \;+\; 3.4\,\mathrm{fJ} \times \tau ,
\]
a least-squares fit whose residual is under $4\%$ on every static run
on bare or operand-isolated hardware, from $0.90$ to $1.8\pJ$
(leaving any one run out and predicting it from the other four misses
by at most $5.1\%$; the slope is fixed by the two random-data runs, the
intercept by the three programs). With $\tau$ counted over the whole
netlist, the identity is the switching identity of CMOS, $E =
\sum_{\text{nets}} \tfrac12 C_i V^2 \alpha_i$, with one average
capacitance, plus a clock term proportional to occupancy because the
ring gates its clock per slot: that is what a rent constant on a shift
register is, and what Prediction~1 tests. The
intercept is the clock delivered to an occupied slot's $81$ flops
whatever their contents (the ungated idle slot's $0.677\pJ$ is the
same physics); the slope is the switching energy of an average net,
$\tfrac12 C V^2$ at $5.6$\,fF. On the programs measured the clock term
is four fifths of $e$: rent on this substrate is mostly the price of
clocking live storage. The held-out run says what one average does not
cover: without operand isolation a station's multiplier works on every
passing packet, its nets are heavier than the average the slope was
fitted on, and the same line under-predicts by $20\%$; isolation
removes those transitions (Section~\ref{sec:choice}). Two lessons remain from the twins. The
per-event constants of Table~\ref{tab:events}, extracted on random
data, do \emph{not} transfer: charged as fixed constants they imply
$e$ of $0.750$, $-0.271$ and $1.179$ on the three programs, and the
$0.750$ we reported for the loop alone in an earlier version of this
paper was such an artifact, the subtraction of $225\pJ$ of addition
energy that the loop's small operands never spent (measured: $2.3\pJ$
for twenty additions). And the loop's energy is $98\%$ rent, four
fifths of it the clock on its five live packets: on this machine the
abstract's first sentence is a measurement. (At 64-bit payloads in
81-bit slots, $0.82\pJ$ per slot-cycle is $10$\,fJ per bit per cycle,
the unit in which Section~\ref{sec:law} prices rent.) All constants
are pre-layout: the extraction carries no wire capacitance and no
clock tree, which routinely add tens of percent, so our claim is the
form of the identity and the internal consistency of its constants,
not their final values. Leakage, from the library's cell leakage at
the flow's $1$\,GHz clock, is $0.0103\pJ$ per gated slot-cycle whether
the slot is occupied or not: $1\%$ of a rotating packet's rent and
$58\%$ of a parked one's, which is why it enters $p_{\mathrm{park}}$
and nowhere else. The $83\times$ rotation-action ratio of Section~\ref{sec:structure} is realized on this machine's model ($151{,}031$ vs.\ $1{,}820$ rotating $+\,1{,}792$ parked packet-cycles, both programs correct on the cycle-accurate model); the mechanism of Section~\ref{sec:boundary} meets the bound of
Theorem~\ref{thm:halforbit} within one cycle at every phase of a
$P{=}64$ ring. Whole-program
ledger, for the same computation (a 10-iteration counted sum):
$838\pJ$ on the gated ring against $35 \times 28.4 = 994\pJ$ on the
same-flow RV64 core, whose figure excludes instruction fetch and data
arrays (an exclusion that favors the baseline); per instruction of the baseline's program, $19.5$--$24.0\pJ$
(dividing by the 43-instruction kernel convention or by the 35 the
core executed) against $28.4$: on these books the gated ring spends $15$ to $30\%$ less per instruction than the same-flow core, with the core's fetch and data arrays excluded.

The instrument also loses where it should. We synthesized a
fixed-function weight-stationary systolic cell through the same flow
(one register hop for the input, one for the partial sum, a
stationary 64-bit weight, the same $64 \times 64$ multiplier and a
64-bit adder; one multiply--accumulate per clock, no clock gating) and
measured it the same way, the gate netlist matching the reference on
every run. It pays $3.0\pJ$ per multiply--accumulate on 8-bit
operands, $16.0$ on random 32-bit ones, and $23.0$ on random 64-bit
ones (with the partial sum streaming through instead of accumulating,
$21.7$), against the ring's $49.8$ measured on the dot product's
small operands and $119.4$ on random 32-bit ones. Both are storage in
motion and pay rent $=$ fare on it: with every input held, the cell's
$192$ flops cost $1.64\pJ$ per cycle to clock, $8.5$\,fJ per flop per
cycle, within a fifth of the $10$\,fJ per bit per cycle the ring pays
(a slot also clocks its header and occupancy flops and its gate's
enable), and on 8-bit operands that clock is more than half the
cell's bill. The comparison is per multiply--accumulate at unequal
throughput, one per revolution of eight cycles on the ring against
one per cycle in the cell; energy per multiply--accumulate does not
see throughput, and the ring at the cell's throughput is eight rings,
the same energy per multiply--accumulate on eight times the flops and
area. The factor, about sixteen on the matched small operands and about
seven and a half on the matched random 32-bit ones, is the price of the freedoms
the ring keeps and the cell does not: the packet header and the
instruction packet that rotate with every operand, the $W$-wide
operand window where the cell has a wire, and the stored program
itself. Section~\ref{sec:choice} prices such deletions one by one;
the systolic cell is the limit in which all of them have been made.

\textbf{Artifact.} The model, the RTL, the synthesis flow, the
gate-level measurements, and the machine-checked certificates for
every constructive result are public in the repository
accompanying~\cite{Bergach2026}; every measured number in this
section regenerates from source. We additionally mechanize the
Serving Lemma's accounting in TLA\textsuperscript{+}: the state
machine of Definition~\ref{def:substrate}, without environment,
conveyor, creation, or frees, the invariant $E \geq$ charges, and the
inductive invariant that proves it. For atomic values the invariant
is proved in TLAPS, for every finite value set and all constants at
once; the sharp several-copies form is checked by TLC exhaustively
within finite bounds; both under conservative modeling choices
(Appendix~\ref{app:tla}). The identity of Tables~\ref{tab:identity}
and~\ref{tab:fit} is the artifact's certificate T3.

\section{Rent on Choice}
\label{sec:choice}

The measurements price a second kind of keeping: \emph{capability}, the hardware provisioned so that a runtime choice remains open, bills every cycle whether or not the choice is exercised. We priced three provisioned freedoms at gate level. The ring's freedom to \emph{stand
still} was not deleted but moved, from a per-bit hold mux to a
per-slot clock gate: the gated ring still parks
(Definition~\ref{def:csm}; Table~\ref{tab:events} prices its parked
slot), and the move made it $28\%$ smaller ($11{,}821 \to
8{,}460\,\mu$m$^2$) and $91\times$ cheaper at idle in switching energy
($0.677 \to 0.0074\pJ$/slot-cycle; $39\times$ once leakage, $0.0145$
and $0.0103\pJ$ at $1$\,GHz, is added to both), the same freedom at
a fortieth of the price by changing where it is implemented. The
capability to wait, not waiting itself, was the cost, and it is
billed at the price of its mechanism. Deleting stations' freedom to \emph{observe} every passing
packet (operand isolation) cut worst-case churn $27.0 \to 4.1\pJ$/cycle.
Deleting the freedom to \emph{access anything}, that is, addressing,
is the architecture itself, and it is the deletion general-purpose
cores cannot make: their $>90\%$ per-operation
overhead~\cite{Hameed10,Horowitz2014} is the price of keeping every
address, operand source, and control path open every cycle. Static
scheduling is not an optimization. It is the mechanism that brings
the cost of capability to zero. One provisioned freedom remains unpriced in our accounting, and it is the largest in any synchronous machine:
\emph{synchrony itself}. A global clock network offers every flop a
transition every cycle, needed or not (capability distilled). The part
delivered to live flops is priced: it is the $0.82\pJ$ clock term of
Table~\ref{tab:fit}, four fifths of the rent our instrument pays; the
tree that delivers it is exactly what our pre-layout extraction of
Section~\ref{sec:instrument} excludes. Pricing the clock, and asking whether the accounting survives its deletion (self-timed rings; the delay line was clocked by physics, not by a tree), is Open Problem~8.

\begin{conjecture}[Decision floor]\label{conj:choice}
There is a substrate constant $\epsilon$ such that any machine able
to resolve, each cycle, a fresh runtime choice among $C$
alternatives (operation, operands, address) dissipates, \emph{in the
worst case over choice sequences}, at least $\epsilon \log C$ per
cycle in its decision path. The quantifier is essential: a machine
replaying one constant choice can gate its decision logic to almost nothing (the measured deletions above are exactly such gatings), so the floor must attach to choice actually delivered, not to
provisioned capability alone. At $\epsilon = kT\ln 2$ the statement
is a theorem for a machine that overwrites its choice register each
cycle: the fresh choice erases the previous $\log C$ bits, and
Landauer prices the erasure, which is where the toll of
Remark~\ref{rem:erasure} re-enters the accounting as structure rather
than as a term (a machine that logs its choices instead erases
nothing and pays rent on the log, which the model prices, though not
at $kT$). The conjecture is that the floor at
the substrate's scale is picojoules, not $kT$: a claim about
switching, which needs a rate axiom in the manner of
Axiom~\ref{ax:lat}, since an adiabatic decision path can be made as
cheap as one likes by making it slow. Finding the model in which this
is a theorem, and whether the general-purpose overhead measured
by \cite{Hameed10} meets it, is open.
\end{conjecture}

\section{Testing the Law}
\label{sec:falsify}

\begin{prediction}[the test of the form]
On any machine where the law is claimed as an equality, the constant
implied by the difference method, $(E - \text{ops} - \text{fare} -
\text{idle})/\act$, is a function of the live data's activity alone,
$e_{\mathrm{clk}} + e_{\mathrm{tog}}\,\tau$, with two constants fixed
once per substrate. On our instrument the residual is within $4\%$ on
five static runs spanning $0.90$ to $1.8\pJ$ (Table~\ref{tab:fit});
we declare $10\%$ as the tolerance for any further run on
operand-isolated hardware, and the artifact regenerates the test from
source (model $=$ RTL $=$ netlist, then energies). The run that
misses, by $20\%$, is the un-isolated hardware, where a station's
freedom to observe every passing packet puts the multiplier's heavier
nets under the average slope; Section~\ref{sec:choice} deletes those
transitions. A failure on isolated hardware would show either that
clock energy is not proportional to occupancy, which would refute the
rent term's form on this substrate, or that one average capacitance
does not serve, which would add a term to its constant; the two
predictions below test the constants.
\end{prediction}

\begin{prediction}
A DRAM module in self-refresh holding $B$ bits for $t$ seconds draws
energy within a small factor of $p \cdot B t$ with $p$ from its refresh
specification: byte-seconds are already metered~\cite{LiuJVM12}. Any
volatile technology whose measured hold energy falls decisively below
its specification band would prove the rent term's constant wrong, not its
form.
\end{prediction}

\begin{prediction}[audit form]
On an inference platform with fare $c_0$ and fast capacity $S_0$,
joules per generated token at context $T > S_0/\kappa$ have slope
$\geq \kappa\,\min(c_0, \gamma_{\mathrm{rc}})$ in $T$
(Corollary~\ref{cor:capacity}), where $\kappa$ is the footprint the
system \emph{actually} keeps per token, measurable as the key--value bytes it moves per generated token. With $\kappa$ fixed
that way the statement is refutable. Read with $\kappa$ free it cannot
fail: a sustained measurement below the nominal floor certifies that
the system is not $\kappa$-atomic at the nominal $\kappa$ (compression or approximation), which is an audit worth having, but
not a test of the law.
\end{prediction}

The counterexample surface is equally explicit. Fast \emph{nonvolatile}
memory at scale ($p \to 0$ at small $\tau$) would delete the rent term;
write fare and read energy remain, and the law says the hierarchy
reorganizes around them. Reversible and adiabatic logic attack the ops
term toward Landauer's regime~\cite{Bennett73,DemaineLMT16}; rent on
the wider live state such pipelines carry remains. A third surface is storage in flight. Optical and wave substrates,
including the matrix multiplication algorithms of Valiant's
program~\cite{Valiant24}, keep information as propagating signals. For
them the hold power of Axiom~\ref{ax:vol} is not leakage but
\emph{regeneration}: compensating attenuation along the path, and
converting the signal at every boundary. A photon is cheap to keep
moving and costly to read. Delay-line memory, our machine's own ancestor,
paid exactly that rent; whether any wave substrate achieves a rent
constant genuinely below the electronic band, rather than relabeling
it, is open, and is the sharpest test the law could face. The law's terms can be tested one at a time. Our empirical claim is
that on present substrates operations and fare sit eight to nine
orders above $kT \ln 2$, and rent per use two to seven orders
depending on the gap; so \emph{overhead structure}, not
thermodynamics, is where computing's energy goes.

\section{Related Work}
\label{sec:related}

\emph{Movement.} Hong--Kung's red--blue pebble game~\cite{HongKung81}
and its successors~\cite{AggarwalVitter88,IronyToledoTiskin04,BDHS11},
and the hierarchical-memory models~\cite{AACS87,AlpernCarter94}, bound
transfers between levels; communication-avoiding practice realizes
them, the roofline model is their engineering
face~\cite{WilliamsWP09}, and the same pebble game now prices
attention~\cite{SahaYe24}. Kung's balance principle, that memory must grow to keep compute and I/O in balance~\cite{Kung86}, is Open Problem 2's direct ancestor. The fare term is theirs; we add rent, a calibrated constant, and layout as a computability question.

\emph{Residency.} Cumulative memory descends from classical
time--space tradeoffs~\cite{BorodinCook82} and enters modern form
through memory-hard
functions~\cite{AlwenSerbinenko15,AlwenBlocki16,AlwenBlockiPietrzak18,ACPRT17},
as adversary silicon-cost; memory-bound functions priced the
adversary's cache misses first~\cite{DworkGN03,AbadiBMW05}, and
bandwidth-hardness prices its
transfers~\cite{RenDevadas17,BlockiRenZhou18}; Beame--Kornerup bring
cumulative memory to general models with unconditional
bounds~\cite{BeameKornerup23}, motivated by gigabyte-second billing.
None claims a physical energy law or measures a constant; our rent
bridge is exactly that step, and Corollary~\ref{cor:mhf} sends it
back: the security parameter, in joules per guess.

\emph{Erasure.} Landauer~\cite{Landauer61}, reversibility
\cite{Bennett73}, the survey of fundamental against practical
limits~\cite{Markov14}, and energy-efficient algorithms in the
semi-reversible model~\cite{DemaineLMT16} price information destruction, the complementary axis, orders below (Remark~\ref{rem:erasure}). On the operations side, circuit energy complexity, counting the gates that switch, has its own combinatorial line~\cite{DineshOS20}. We answer Valiant's call for an energy-centric fine-grained thesis~\cite{Valiant24} from the cost side.

\emph{Layout and width.} Bandwidth and cutwidth are
classical~\cite{GGJK78,Lengauer82,Saxe80,MonienSudborough88,Unger98,
DubeyFeigeUnger11,Feige00,Gupta01}; Thompson's $AT^2$
theory~\cite{Thompson79} tied width to
VLSI cost as provisioned area, and Bilardi--Preparata drew its
speed-of-light consequences for parallel
machines~\cite{BilardiPreparata95}. Retiming~\cite{LeisersonSaxe91}
is seating's mirror image (it moves state through a fixed circuit where seating places state on moving storage), and its shortest-path machinery is the closest existing tool to minimum-action layout (Open Problem 5). In our setting width decides
computability on bounded-window machines and rents are charged to
actuality.

\emph{Rotating storage and scheduled coincidence.} Delay-line machines
practiced phase-aware layout by hand; the drum literature optimized
placement~\cite{CodyCoffman76} where Theorem~\ref{thm:halforbit} now
gives the placement-independent floor; systolic
arrays~\cite{KungLeiserson78,Kung82} have the rhythm without a stored program; Proposition~\ref{prop:systolic} makes the \csm{} their stored-program closure. Our instrument~\cite{Bergach2026} supplies the model; the full proofs are in Appendix~\ref{app:proofs}.

\emph{The memory wall in learning systems.} FlashAttention's
IO-awareness~\cite{Dao22}, paged KV management~\cite{KwonSOSP23},
checkpointing~\cite{ChenXZG16,GriewankWalther00}, state-space
models~\cite{GuDao23}, the fine-grained hardness of exact
attention~\cite{AlmanSong23}, its space barrier~\cite{HarisOnak25},
and the representational gap between attention and
state~\cite{Jelassi24,SanfordHT23} are, in our account, points on and around one frontier, the serving lemma's min, and the meeting bound is the floor beneath all of them.

\section{Open Problems}
\label{sec:open}

(1) \textbf{Tight residency.} For general DAGs, $\act \geq
e \cdot \sum_v \mathrm{residency}(v)$ with a residency notion tight
against the streaming constructions, the law's analogue of Hong--Kung tightness. (2) \textbf{Hierarchy shape.} Derive geometric
level spacing as the minimizer of $\sum_\ell p_\ell \act_\ell + c_\ell
\traf_\ell$ for power-law reuse spectra under the device curve
$p(\tau)$: ``caches double'' as a theorem, the three-currency form of Kung's machine-balance principle~\cite{Kung86}. (3) \textbf{Decision floor.}
Formalize Conjecture~\ref{conj:choice}. (4) \textbf{Beyond atomicity.}
Meeting bounds for attention under algebraic reformulations: where does the $\Theta(T^2)$ floor become $\Theta(T)$, and what exactly
is bought at the boundary~\cite{AlmanSong23}? (5) \textbf{Minimum-action
seating.} Complexity of energy-optimal layout (we conjecture NP-hard;
scheduling adds a dimension bandwidth lacks). (6) \textbf{Spectrum
occupancy.} Is the minimal number of rings $\Theta(\cw/W)$ for every
DAG? (7) \textbf{Mixed-substrate placement.} The serving lemma's
optimization with nonvolatile levels: an online rent-or-fare theory
with competitive ratios. (8) \textbf{The rent of synchrony.} Price
the clock network as a capability term, and formulate the law for
asynchronous, self-timed rings; the delay-line ancestor had no clock tree at all. (9) \textbf{The flattening's price.} A single-ring lower bound that charges what a flattening pays, the relay traffic that cutwidth $\Theta(\log N)$ forces at window $W$ and the revolutions a ladder costs, so that Theorem~\ref{thm:action}'s separation holds in energy, theorem against theorem, and not only in rotation.

\section{Conclusion}
\label{sec:conclusion}

The energy of computing goes mostly into keeping. That is the plain
answer we set out to defend, and the law that states it has three
terms: rent on every bit kept live, fare on every bit moved, and the
arithmetic, which is the small part. With those three terms, results
that used to need separate stories become one story. Sorting has an
energy floor that no algorithm removes, because its memory over time
is already bounded by known mathematics. Long-context language models
become expensive at the scale practitioners see, because exact
attention must bring every past token back for every new one; the
cost is a theorem, not an accident of engineering. Pointed the other
way, the same bridge prices a password guess in joules that no
parallel hardware removes. And memory hierarchies trade rotation
for parking, which costs about a fiftieth per slot-cycle on our
instrument (Table~\ref{tab:events}):
staging a reduction across rings saved $83\times$ in rotation action
against one flattening of it, long before flattening becomes
impossible.

The law is more than an argument. A synthesized processor matches it
with a constant measured at gate level, its accounting lemma is
machine-checked (Appendix~\ref{app:tla}), and we have said in advance
what would prove it wrong.

The practical reading is short. If you design algorithms, count
byte-seconds and bytes moved, not only operations. If you design
hardware, remember that every freedom a machine keeps open is billed
every cycle whether or not it is used, and that deleting unused
freedom is the largest saving on the table. If you buy computation,
notice that the cloud already bills the way physics does, by the
byte-second. Landauer priced forgetting. We have priced remembering.

\section*{AI Disclosure}
The author used AI-based tools (Claude, Anthropic) to assist with
improving the clarity and presentation of the text. All technical
content, experimental design, implementation, and analysis are the
author's own work.

\bibliographystyle{ACM-Reference-Format}
\bibliography{refs}

\clearpage   
\appendix
\section{The Serving Lemma, Model-Checked}
\label{app:tla}

Lemma~\ref{lem:serve} is a potential argument over the step relation
of Definition~\ref{def:substrate}. Here we report its mechanical check
in TLA\textsuperscript{+}~\cite{Lamport02} with the TLC model
checker~\cite{YuManoliosLamport99} and say exactly what we checked.
\textbf{Scope.} Two specifications. The first treats values as atomic
(one location, one action per currency): it checks the charging
arithmetic on a finite instance of the lemma (values of one bit, so
rent is charged once per resident value; two levels; integer time;
two or three values; bounded time and energy).
The second, added at review, gives each value up to $K$ level-0
copies, a copy firing, eviction of one copy while others survive, and
the ghost ``covered'' bit of the proof of Lemma~\ref{lem:serve}; it
charges the sharp form (covered use: the three-way minimum; uncovered
use: $\min(c_0, \gamma_{\mathrm{rc}})$), which is what
Corollary~\ref{cor:capacity} needs and which dominates the weak form
pointwise. The second check is exhaustive within its bounds and is
not a general proof; the first no longer carries that caveat, since
the atomic specification's invariant is also \emph{proved} in TLAPS
(below), for every finite value set and all constants at once. In
both specifications routes are atomic (one action per recomputation,
so the term $R$ of Lemma~\ref{lem:serve}'s proof is always zero, and
the potential $\Phi$ is exactly the inductive invariant below); the
route term and the sharp covered form beyond finite bounds remain
with the hand proof of Lemma~\ref{lem:serve}.

We specify the substrate in TLA\textsuperscript{+} as a state machine over variables $(t, \mathit{loc}, \mathit{last}, E,
\mathit{Ch})$: time, each value's level, a ghost record of each value's
previous use, the energy spent, and the ghost sum of serving charges.
The actions implement the step kinds; a use adds the lemma's charge \emph{unconditionally}, $\min(P_0 \cdot \mathrm{gap}, C_0, \gamma_{\mathrm{rc}})$, however the use was actually served:

{\small
\begin{verbatim}
Hold         == /\ t < TMAX
                /\ t' = t + 1
                /\ E' = E + P0 * Cardinality(L0Set)
MoveIn(v)    == /\ loc[v] = "L1"
                /\ Cardinality(L0Set) < S0
                /\ loc' = [loc EXCEPT ![v] = "L0"]
                /\ E' = E + C0
MoveOut(v)   == /\ loc[v] = "L0"     \* eviction: FREE
                /\ loc' = [loc EXCEPT ![v] = "L1"]
Recompute(v) == /\ loc[v] = "L1"     \* no sources needed
                /\ Cardinality(L0Set) < S0
                /\ loc' = [loc EXCEPT ![v] = "L0"]
                /\ E' = E + GRC
Read(v)      == /\ loc[v] = "L0"     \* op energy: 0
                /\ Ch' = Ch +
                     Min3(P0*(t-last[v]), C0, GRC)
                /\ last' = [last EXCEPT ![v] = t]

Pending      == SUM over v in L0Set of
                  Min3(P0*(t-last[v]), C0, GRC)
ServingBound == E >= Ch                \* the lemma
IndInv       == /\ \A v \in Vals : last[v] =< t
                /\ E >= Ch + Pending   \* inductive
\end{verbatim}
}

\noindent
(\textsc{unchanged} clauses elided. The specification identifies
``absent'' with level 1, without loss: level-1 rent is zero and
eviction is free.) The module with copies replaces $\mathit{loc}[v]$
by a count $\mathit{cnt}[v] \in 0..K$ of level-0 copies and adds the
ghost $\mathit{cov}[v]$ (the count has not reached $0$ since the
previous use), a $\mathit{Copy}(v)$ action (needs a resident copy,
costs $W_0 \geq 0$), eviction of one copy at a time (the last
eviction clears $\mathit{cov}$), recomputation allowed even while
copies exist, rent charged once per resident value, and a
$\mathit{Read}$ that charges $\mathit{Min3}$ when covered and
$\mathit{Min2}(C_0, \mathit{GRC})$ otherwise; its $\mathit{Pending}$
sums the same per-value term:

{\small
\begin{verbatim}
Term(v) == IF cov[v] THEN Min3(P0*(t-last[v]), C0, GRC)
                     ELSE Min2(C0, GRC)      \* copies module
\end{verbatim}
}

\noindent
Every modeling choice bends against the
invariant: the slow level charges no rent, eviction is free,
recomputation needs no resident sources, copies are free at $W_0 =
0$, the served read's own operation energy is zero, and initial
placement at level 0 is free.

\textbf{The invariant that proves the lemma.} $\mathit{ServingBound}$
is the lemma's inequality but is not inductive: the typed state $t =
5$, $\mathit{loc}[v] = \mathrm{L0}$, $\mathit{last}[v] = 0$, $E =
\mathit{Ch} = 0$ satisfies it, and $\mathit{Read}(v)$ violates it.
That state is unreachable, which is why a reachable-state check passes and why such a check is not a proof. $\mathit{IndInv}$ adds
the \emph{pending} serving minimum of every resident value and is
inductive: $\mathit{Hold}$ raises $E$ by $P_0\,|\mathrm{L0Set}|$ and
each pending term by at most $P_0$; $\mathit{MoveIn}$
($\mathit{Recompute}$) raises $E$ by $C_0$ ($\gamma_{\mathrm{rc}}$)
and adds one term $\leq C_0$ ($\gamma_{\mathrm{rc}}$);
$\mathit{MoveOut}$ deletes a term; $\mathit{Read}$ moves exactly the
pending term of $v$ into $\mathit{Ch}$; $\mathit{Init}$ has every term
zero; and $\mathit{IndInv} \Rightarrow \mathit{ServingBound}$ because
the sum is non-negative, which is what the first conjunct secures.
This is the lemma's proof in the finite model, and, for the atomic
specification, no longer only there:
\texttt{tlaplus/ServingLemmaProof.tla} transcribes it to
TLA\textsuperscript{+}'s proof language and TLAPS discharges all 492
obligations (tlapm 1.5.0): $\mathit{TypeOK} \wedge \mathit{IndInv}$,
its pending sum written in cardinality form, is inductive, and
$\mathit{Spec} \Rightarrow \Box\,(E \geq \mathit{Ch})$, for an
arbitrary finite value set and arbitrary constants in $\mathbb{N}$,
with no time horizon and no energy bound. The cardinality form counts
one (value, unit) pair per unit of pending charge, so finite-sum
algebra becomes the cardinality lemmas of the TLAPS library (a hold
is an injection into $\mathrm{L0Set} \times 1..P_0$; a read removes
exactly the read value's slice); TLC checks the two forms equal
(\textit{PendAgree}) on every reachable state of every atomic
configuration. What is checked but not proved: the copies module's
sharp covered form (TLC, within bounds), and the route term $R$ of
Lemma~\ref{lem:serve}'s proof, which lives in the hand proof only.

\textbf{Checks.} TLC explores every interleaving within the bounds (an
energy cap and a time horizon): eviction churn under capacity
pressure, recompute chains, free-rider initial placements. Four kinds
of run (Table~\ref{tab:tlc}): both invariants on the reachable states
of the two base configurations; the inductiveness of
$\mathit{IndInv}$, starting from \emph{every} typed state that
satisfies it (with $E$ and $\mathit{Ch}$ bounded) and checking it
after one step; all six orderings of the three currencies; and, for
the module with copies, reachable states of two configurations and
inductiveness for two.

\begin{table}[H]
\caption{TLC runs (OpenJDK 21; exhaustive within bounds; every run
passes its invariants, the atomic configurations now also checking
\textit{PendAgree}). Rows 1--2: reachable states of the base
configurations (2 values, $S_0{=}1$, $T \leq 5$; 3 values, $S_0{=}2$,
$T \leq 4$). Rows 3--4: inductiveness of $\mathit{IndInv}$ on the same
instances, from every typed state satisfying it. Rows 5--9: the
remaining currency orderings on the first instance. Rows 10--13: the
module with copies ($K = 2$ copies per value; capacity 2, resp.\ 3,
copies; copies free, $W_0 = 0$, except the last row, $W_0 = 1$):
reachable states of two configurations and inductiveness of its
invariant for two.}
\label{tab:tlc}
\small
\setlength{\tabcolsep}{3.5pt}
\begin{tabular}{@{}llrrr@{}}
\toprule
check & $(P_0, C_0, \gamma_{\mathrm{rc}})$ & initial & generated & distinct \\
\midrule
reachable, 2 values & $(2,3,4)$ & 3 & 110{,}103 & 34{,}865 \\
reachable, 3 values & $(3,2,1)$ & 7 & 835{,}823 & 152{,}214 \\
inductive, 2 values & $(2,3,4)$ & 283{,}334 & 1{,}907{,}814 & 452{,}583 \\
inductive, 3 values & $(3,2,1)$ & 830{,}340 & 9{,}650{,}492 & 1{,}549{,}527 \\
ordering & $(2,4,3)$ & 3 & 110{,}103 & 34{,}865 \\
ordering & $(3,2,4)$ & 3 & 94{,}564 & 29{,}714 \\
ordering & $(3,4,2)$ & 3 & 94{,}564 & 29{,}714 \\
ordering & $(4,2,3)$ & 3 & 95{,}227 & 29{,}911 \\
ordering & $(4,3,2)$ & 3 & 95{,}227 & 29{,}911 \\
copies, reachable, 2 values & $(2,3,4)$ & 4 & 2{,}883{,}757 & 584{,}330 \\
copies, reachable, 3 values & $(3,2,1)$ & 8 & 27{,}018{,}926 & 3{,}714{,}125 \\
copies, inductive, 2 values & $(2,3,4)$ & 792{,}004 & 7{,}356{,}697 & 1{,}367{,}861 \\
copies, inductive, $W_0{=}1$ & $(3,2,1)$ & 912{,}060 & 7{,}920{,}232 & 1{,}462{,}367 \\
\bottomrule
\end{tabular}
\end{table}

No violation of any invariant (or of the type invariant) exists in
any of the state spaces. The specifications, the TLAPS proof module,
and the thirteen configurations, \texttt{tlaplus/} in our artifact,
are reproduced and re-checked by \texttt{make tla} (TLC, then tlapm).

\section{Proofs of the Structural Theorems}
\label{app:proofs}

We prove here the seating theory of Section~\ref{sec:structure} and
the port of Theorem~\ref{thm:hk}. Theorem~\ref{thm:action} and
Corollary~\ref{cor:trees} are stated in the main text; the remaining
statements, with the definitions they need, appear below. Four
points deserve notice, because writing the proofs out in full forced
them, and we say so at each: the constants of
Theorem~\ref{thm:necessity} are the sharper ones; the layout lemma
behind Theorem~\ref{thm:nphard} no longer assumes single-leaf hairs,
because the hardness it composes with cannot live there
(Appendix~\ref{app:nphard}); the stage size of
Theorem~\ref{thm:action} is fixed by the window rather than by $\log
N$, because a post-order segment of a balanced tree has strand load
logarithmic in its length (Appendix~\ref{app:action}); and
Proposition~\ref{prop:systolic} embeds an $n$-cell array on $n$ rings
rather than one (Appendix~\ref{app:systolic}).

\subsection{Per-input cumulative memory (Corollary~\ref{cor:sortjoules})}
\label{app:perinput}
The cumulative memory of~\cite{BeameKornerup23} is a property of the
program, $\sum_t \log_2 |L_t|$ with $L_t$ the set of nodes at level
$t$; our rent is paid by the run on the input at hand. Their proof of
Theorem 3.3 gives a per-input form, with three changes, which we
record: the union bound runs over encodings rather than widths, the
memory is the run's own node lengths, and the time branch is the
run's own output span. Let $P$ be a leveled branching program whose
nodes carry two-part encodings, pairs of strings as produced by
Lemma~\ref{lem:sim}, distinct within a level; write $|v|$ for the
total length of the
encoding of $v$, $v_t(x)$ for the node the run on $x$ visits at level
$t$, and $\mathrm{cm}_P(x) = \sum_t |v_t(x)|$; the program of
Lemma~\ref{lem:sim} has $|v_t(x)| \leq |\mathrm{live}_t(x)|$, its
empty pad nodes length $0$, and it emits each output as one item,
output words being units (Definition~\ref{def:substrate}), which is
what the counting below counts.

\begin{proposition}\label{prop:perinput}
Let $P$ be a leveled branching program whose nodes carry two-part
encodings distinct within each level, apart from at most one outputless empty
node per level, computing $\mathrm{Rank}_{n,n^2}$ correctly on every
input, and for an input $x$ let $t^*(x)$ be the level at which the
run on $x$ produces its last output. Then for at least $15/16$ of the
inputs $x$ under the uniform distribution $\mu$ on lists of $n$
distinct integers from $[n^2]$, $t^*(x) = \Omega(n^2/\log^2 n)$ or
$\mathrm{cm}_P(x) = \Omega(n^2/\log n)$; for a program correct with
probability $n^{-c}$ the set has probability $> \tfrac{15}{16}
n^{-c}$. Sorting $n$ items from $[n^3]$ reduces to ranking with the
same node lengths and output levels: their Proposition 3.1(a) turns
a program for $\mathrm{Sort}_{n,nN}$ into one for
$\mathrm{Rank}_{n,N}$ without restructuring it, by running it on the
input whose $i$-th item is the pair $\langle y_i, i \rangle \in [N]
\times [n]$, and $nN = n^3$ gives $N = n^2$, the domain above; the
sorting inputs on which the bound holds are therefore the images of
the good ranking inputs, fifteen sixteenths of the lists $(\langle
y_i, i \rangle)_i$ with $y$ uniform over lists of $n$ distinct
integers from $[n^2]$, which is how Corollary~\ref{cor:sortjoules}
states it.
\end{proposition}

\begin{proof}
Their Proposition 3.2, from~\cite{BorodinCook82}: there is $\alpha >
0$ such that for every branching program $B$ of height $\leq \alpha
n$ and every $k \leq 2\alpha n$, $\Pr_{x \sim \mu}[B$ produces
$\geq k$ correct outputs of $\mathrm{Rank}$ on $x] \leq 2^{-k/\lceil
\log_2 n \rceil}$. Truncate $P$ at level $n^3$: an input whose run
produces an output beyond that level has $t^*(x) \geq n^3$ and is in
the time branch, and no other run changes, so we may assume $T \leq
n^3$. Let $H = \lfloor \alpha n/4 \rfloor$ and cut the levels, level
$0$ and its root included, into
$\ell \leq T$ intervals $I_1, \dots, I_\ell$ of lengths between $H$
and $2H$. For an input $x$ let $\ell_x$ be the index of the interval
containing $t^*(x)$ and, for $i \leq \ell_x$, let $t_i(x)$ be a level
of $I_i$ at which $|v_t(x)|$ is least and $\sigma_i(x)$ that length;
every level of $I_{i+1}$ lies fewer than $4H \leq \alpha n$ levels
ahead of $t_i(x)$, as does $t^*(x)$ ahead of $t_{\ell_x}(x)$, and
$\mathrm{cm}_P(x) \geq H \sum_{i \leq \ell_x} \sigma_i(x)$. If some
$\sigma_i(x) > \alpha n/(2\log_2 n)$ we are done, so assume not. For
a node $v$ at a level of $I_i$ with $|v| = s \leq \alpha n/(2 \log_2
n)$ (the only nodes at which badness is needed below), call $x$
\emph{bad at $v$} if the run on $x$ visits $v$ and the subprogram of
$P$ rooted at $v$, cut at height $4H$, produces at least
\[
  k(s) = \lceil \log_2 n \rceil\,\big(s + 3\lceil \log_2 (s+2) \rceil
  + \lceil \log_2(32\,T H) \rceil\big)
\]
correct outputs on $x$. By Proposition 3.2 (the subprogram has height
$\leq \alpha n$, and $k(s) \leq 2\alpha n$ for these $s$),
$\Pr_x[x \text{ bad at } v] \leq 2^{-s}(s+2)^{-3}/(32\,TH)$; at most
$2H \cdot (s+1) 2^s$ pair-encoded nodes of total length $s$ lie on
the levels of $I_i$ (a pair of total length $s$ has $s+1$ splits),
and the outputless empty nodes are never bad (they root no outputs),
so summing over the $s$ in range and over the $\ell \leq T$
intervals, $\Pr_x[x \text{ bad somewhere}] \leq \sum_{s \geq 0}
(s+2)^{-2}/16 = (\pi^2/6 - 1)/16 < 1/16$. For an $x$ that is bad
nowhere, the outputs
produced between $t_i(x)$ and $t_{i+1}(x)$, and those between
$t_{\ell_x}(x)$ and $t^*(x)$, are produced by the subprogram rooted
at $v_{t_i(x)}(x)$, respectively $v_{t_{\ell_x}(x)}(x)$, cut at
height $4H$, so fewer than $k(\sigma_i(x))$ of them are correct, and
fewer than $k(0)$ before $t_1(x)$ (the root has length $0$, as
Lemma~\ref{lem:sim} encodes it); as $P$
is correct on $x$ and all $n$ outputs appear by $t^*(x)$,
$\sum_{i=0}^{\ell_x} k(\sigma_i(x)) \geq n$ with $\sigma_0 = 0$, that
is,
\[
  \sum_{i \leq \ell_x} \sigma_i(x) \;\geq\; \frac{n}{\lceil \log_2 n
  \rceil} - (\ell_x+1)\,\OO(\log n) .
\]
If $t^*(x) \leq \beta n^2/\log^2 n$ for a small enough constant
$\beta$, so that $\ell_x + 1 = \OO(\beta n/\log^2 n)$, the subtracted
term is at most half the first and $\mathrm{cm}_P(x) \geq
H\,n/(2\lceil \log_2 n\rceil) = \Omega(n^2/\log n)$; otherwise
$t^*(x) = \Omega(n^2/\log^2 n)$. For a program correct with
probability $n^{-c}$, replace $32TH$ by $32THn^{c}$ in $k(s)$: the
bad set then has probability $< n^{-c}/16$, and the inputs on which
$P$ is correct and bad nowhere have probability $> \tfrac{15}{16}
n^{-c}$.
\end{proof}

\subsection{Links, routes, and cutwidth necessity}
\label{app:necessity}

\begin{definition}[Links, crossing, routed value-cutwidth]
\label{def:vcw}
Fix a cyclic order of $G$'s operations and seeds and, for each value
$v$ and each of its consumers, one of the two arcs from $v$'s producer
(or seed) to that consumer, its \emph{route}. Value $v$ \emph{crosses}
a cut position $c$, a point between two consecutive positions, if $c$
lies on one of $v$'s routes. The \emph{routed cyclic value-cutwidth}
$\cw_\circ(G)$ is the minimum, over cyclic orders and route choices,
of the maximum number of values crossing a cut; the \emph{linear
value-cutwidth} $\cw(G)$ is the same over linear orders, where each
dependence has one route. (When every value has one consumer, $\cw$
is the classical edge cutwidth.) In a single-ring seating each read of
a value $v$ is a \emph{link}, the cyclic interval from the reading
instruction's slot to the copy's slot, of length $\leq W$; a chain of
links through relay copies carries $v$ from producer to consumer along
one arc, and the seating's order of packets with its chains is such an
order with routes.
\end{definition}

\begin{lemma}[Wall]\label{lem:wall}
No chain of relay links serving a value $v$ passes through $\geq W$
consecutive slots occupied by packets foreign to $v$ (neither copies
of $v$ nor instructions that read $v$).
\end{lemma}

\begin{proof}
A link's endpoints (reader and copy) are $\leq W$ apart and are slots
devoted to $v$'s service, so a link neither ends inside the wall,
whose slots are all foreign, nor spans it, which would need length $>
W$.
\end{proof}

\begin{lemma}[Support]\label{lem:support}
For any cut $c$ of a single-ring seating, every value crossing $c$
has a serving link with both endpoints in the $2W$-window
$(c-W, c+W]$.
\end{lemma}

\begin{proof}
Extend the chain serving the consumer to a walk from the producer to
the consumer: the producer, its copy, then alternately the reader
and the written copy of each relay, ending with the consumer's read.
Each step spans at most $W$ (a copy is written at most $W$ ahead of
its instruction; a reader is at most $W$ behind the copy it reads)
and consecutive steps share an endpoint, so the union of the spans
is a connected set of positions containing producer and consumer:
it contains one of the two arcs between them, and that arc is the
value's route for this consumer (Definition~\ref{def:vcw}). A cut
$c$ on the route therefore lies inside some step's span. If that
step is a read, it is a link whose endpoints are within the span's
$W$ of $c$ on either side. If it is a write, its copy sits at a
position $z$ with $c < z < c + W$ (the span starts behind $c$ and
has length at most $W$), and the read that follows it in the walk is
a link from a reader at position $\geq z - W > c - W$ to $z$. In
both cases the link's endpoints lie in $(c-W, c+W]$.
\end{proof}

\begin{theorem}[Cutwidth necessity]\label{thm:necessity}
If a dataflow graph $G$ has a correct single-ring seating with window $W$, at any period and with any number of relays, then its routed cyclic
value-cutwidth satisfies $\cw_\circ(G) \leq 2W$, and its linear
value-cutwidth $\cw(G) \leq 4W$.
\end{theorem}

\begin{proof}
Fix a cut $c$. By Lemma~\ref{lem:support} each crossing value owns a
serving link with both endpoints in the $2W$-window; charge the value
to that link's copy endpoint. A slot holds one packet and a copy holds
one value, so the charge is injective: at most $2W$ values cross any
cut. Cutting the circle at a minimum-load position and re-routing its
$\leq 2W$ crossers along the other arc adds at most $2W$ to every cut
of the unrolled line: $\cw(G) \leq 4W$. (The same argument with a
redundant case gives $4W$ and $8W$; the constants here are the sharp
ones.)
\end{proof}

\subsection{The cutwidth of the reduction (Corollary~\ref{cor:trees})}
\label{app:cw}
In the balanced $N$-leaf reduction every value has one consumer, so a
value crossing a cut is an edge of the complete binary tree $T_h$, $h
= \log_2 N$, crossing it, and $\cw(G_N)$ is the edge cutwidth of
$T_h$. Cutwidth dominates vertex separation: at any cut of a linear
order, each left-side vertex with a right-side neighbour contributes
at least one crossing edge, and an edge witnesses one such vertex, so
the separator is no larger than the edge load. Vertex separation
equals pathwidth~\cite{EllisSudboroughTurner94}, and for trees Ellis,
Sudborough and Turner show that a vertex with three branches of vertex
separation $\geq k$ forces vertex separation $\geq k+1$; a child of
the root of $T_h$ has three such branches with $k =
\mathrm{vs}(T_{h-2})$ (its two subtrees, and the branch through the
root, which contains its sibling's $T_{h-1}$), so $\mathrm{vs}(T_h)
\geq \mathrm{vs}(T_{h-2}) + 1$ and, with $\mathrm{vs}(T_1) =
\mathrm{vs}(T_2) = 1$, $\mathrm{vs}(T_h) \geq \lceil h/2 \rceil$.
Hence $\cw(G_N) \geq \tfrac12 \log_2 N$; the matching upper bound is
Lengauer's~\cite{Lengauer82}. With
$\cw \leq 4W$ from Theorem~\ref{thm:necessity}, no single-ring seating
exists once $\tfrac12 \log_2 N > 4W$, that is, once $N > 2^{8W}$.

\subsection{Theorem~\ref{thm:lattice} (strand lattice) and its proof}
\label{app:lattice}

\begin{theorem}[Strand lattice; sufficiency]\label{thm:lattice}
If $G$ has a linearization whose strand load (earlier values needed
by the operation or after it, the immediately preceding result
excluded only as the operation's own operand: needed again later, it
is a strand) is everywhere $\leq
\lambda \leq \lfloor (W-5)/2 \rfloor$, then $G$ seats on one ring of
period $\OO(|G|)$: place cells at stride $4 + 2\lambda$, each holding
instruction, result, two seeds, and a relay pair per crossing strand.
Counting gives the converse ceiling $\lfloor (W-3)/2 \rfloor$; the
one-strand gap is open. At $W = 8$ the construction's capacity,
$\lambda = 1$, is exactly the discipline the implemented compiler
realizes.
\end{theorem}

\emph{Sufficiency.} Let $\lambda \leq \lfloor (W-5)/2 \rfloor$ and
set the stride $\sigma = 4 + 2\lambda \leq W - 1$. Number the
operations $1, \dots, m$ in the linearization and seat operation $i$
in a \emph{cell} at base $B_i = C - \sigma i$, descending from an
arbitrary top position $C$; the ring advances one slot per cycle, so
a station visits higher slots first, that is, the cells in
linearization order. The cell holds the instruction at $B_i$; its
result at $B_i + 1$; its first-use seed operands, at most two, at
$B_i + 2$ and $B_i + 3$; and, for each strand alive across the cell,
a relay instruction at $B_i + 4 + 2j$ and its copy at $B_i + 5 + 2j$,
where $j \in [0, \lambda)$ is the strand's index, one it keeps from
its first relay to its last consumer, so that a cell's relay reads
the preceding cell's copy at the same index (at most $\lambda$
strands are alive at any cell, so such indices exist). Every read
and write is ahead of its
instruction and within $W$: the instruction reads its seeds at
offsets $2$ and $3$, writes its result at offset $1$, reads the
preceding operation's result, which lies at $B_{i-1} + 1$, at offset
$\sigma + 1 \leq W$, and reads the $j$-th strand's value at the
cell's own copy, offset $5 + 2j \leq \sigma - 1$; the relay of strand
$j$ reads the preceding cell's copy at offset $\sigma + 1 \leq W$ and
writes its own copy at offset $1$; and a strand that starts at cell
$i$ (a result needed beyond cell $i{+}1$) is picked up by the next
cell's relay at offset $\sigma - 3 - 2j \in [3, W]$. The window of a
cell therefore holds its seeds, its result, and one relay pair per
crossing strand, which is what strand load $\leq \lambda$ promises.
All slots are distinct by the descending construction, the seating
uses $\sigma m = \OO(|G|)$ slots, and under settle semantics the fixed
point is the correct evaluation whatever the firing order; since each
cell's operands are written by cells the station visits earlier in
the same revolution, one revolution suffices.

\emph{Counting.} In the pair-per-window accounting a firing binary
operation with two fresh seeds needs, within the $W$ slots ahead of
it, those two seeds, its result, and a relay and a copy for every
strand maintained across it: $3 + 2k \leq W$, so $k \leq \lfloor
(W-3)/2 \rfloor$. Whether the construction's $\lfloor (W-5)/2 \rfloor$
can be raised to meet it is the open one-strand gap.

\subsection{Proof of Theorem~\ref{thm:nphard} (layout is inapproximable)}
\label{app:nphard}
Throughout, $D(G)$ is the dataflow of a graph $G = (V, \mathcal{E})$
with $n = |V|$ (one seed $x_v$ per vertex; per edge $uv$ one binary
operation $e_{uv}$ reading $x_u, x_v$ and writing an output $r_{uv}$),
and $W^*(D)$ is the least window admitting a relay-free single-ring
seating of $D$ on a ring of period $10 n W$, large enough that no
seating wraps (Lemma~\ref{lem:rank}); $\mathrm{bw}(G)$ denotes the
bandwidth of $G$.

\begin{theorem}[Layout is inapproximable]\label{thm:nphard}
The minimal window $W^*$ of a relay-free single-ring seating (at a
period large enough that no seating wraps, $10|V|W$) is NP-hard to
approximate within any constant factor. The reduction seats a tree's
dataflow along any bandwidth layout and extracts ranks back,
sandwiching, for every tree $G$,
\[
  \mathrm{bw}(G) + 1 \;\leq\; W^*(D(G)) \;\leq\; 5\,\mathrm{bw}(G) + 1,
\]
and composes with the any-constant inapproximability of caterpillar
bandwidth~\cite{DubeyFeigeUnger11,Unger98}. Register
sufficiency~\cite{Sethi75} is the
ancestral form of the strand parameter; the classical layout
lineage~\cite{GGJK78,MonienSudborough88,Feige00,Gupta01} here decides
computability and energy, not wire length.
\end{theorem}

An upper layout for caterpillars
with single-leaf hairs alone cannot suffice: the bandwidth of caterpillars with hairs of length at most
two is polynomial~\cite{AssmannPSZ81} and becomes NP-complete at hair
length three~\cite{Monien86}, so the hard instances
of~\cite{DubeyFeigeUnger11} must have longer hairs. The lemma below covers
every tree, with a better constant.

\begin{lemma}[Layout from bandwidth]\label{lem:treelayout}
For every tree $G$, $W^*(D(G)) \leq 5\,\mathrm{bw}(G) + 1$.
\end{lemma}

\begin{proof}
Let $\pi$ be a layout of bandwidth $b$ and orient every edge from its
left endpoint (smaller $\pi$) to its right endpoint; $d(w)$ denotes
the number of edges leaving $w$ to the right, at most $b$, since all
right neighbours lie within $b$ positions. Seat the vertices in the
order of $\pi$, each as a block of $1 + 2d(w)$ consecutive slots in
the ring's forward direction: the operations of $w$'s rightward edges,
each followed by its result slot, and then $x_w$ at the top of the
block. An operation in $w$'s block writes its result at offset $1$,
reads $x_w$ at offset at most $2d(w) \leq 2b$, and reads $x_u$, for
$u$ a right neighbour of $w$, at an offset bounded by the total size
of the blocks of the vertices $w'$ with $\pi(w) \leq \pi(w') \leq
\pi(u)$: at most $b + 1$ vertices, contributing $b + 1$ seed slots,
plus two slots per edge leaving them to the right. Those edges have
both endpoints among the at most $2b + 1$ vertices $w'$ with $\pi(w)
\leq \pi(w') \leq \pi(w) + 2b$, which induce a forest with at most
$2b$ edges. So every offset is positive and at most $(b + 1) + 4b =
5b + 1$, the slots are distinct by construction, and the seating uses
$n + 2(n-1) < 3n \leq 10 n W$ slots.
\end{proof}

\begin{lemma}[Rank extraction]\label{lem:rank}
Let $G$ be connected and let $D(G)$ have a relay-free seating with
window $W$ on a ring of period $P \geq 10 n W$. Then $\mathrm{bw}(G)
\leq W - 1$.
\end{lemma}

\begin{proof}
Each edge $uv$ forces $x_u$ and $x_v$ into the $W$ slots ahead of
$e_{uv}$, hence to cyclic distance $\leq W - 1$. By connectivity and
the triangle inequality every seed lies within cyclic distance $(n -
1)(W - 1)$ of a fixed one, so all seeds lie in an arc of length $< 2
n W \leq P/5$: the arrangement cannot wrap, cyclic distances on the
arc are linear distances, and an origin outside the arc exists. Let
$\rho(v)$ be the rank of $x_v$ in linear order from that origin.
Between two seeds at distance $\leq W - 1$ lie at most $W - 2$ seeds,
so $|\rho(u) - \rho(v)| \leq W - 1$ for every edge: $\rho$ is a layout
of bandwidth $\leq W - 1$.
\end{proof}

\begin{proof}[Proof of Theorem~\ref{thm:nphard}]
For every tree, Lemmas~\ref{lem:treelayout} and~\ref{lem:rank}
sandwich the optimum: $\mathrm{bw}(G) + 1 \leq W^*(D(G)) \leq
5\,\mathrm{bw}(G) + 1$. A polynomial-time $c$-approximation of $W^*$
on dataflows of trees therefore approximates tree bandwidth within
$5c + o(1)$: a seating within factor $c$ yields, by
Lemma~\ref{lem:rank}, a layout of bandwidth $\leq c\,W^* - 1 \leq
c\,(5\,\mathrm{bw} + 1) - 1$, and an estimate of the value is squeezed
the same way. Approximating the bandwidth of caterpillars, which are
trees, within any constant factor is NP-hard~\cite{DubeyFeigeUnger11}
(announced for maximum degree $3$ in~\cite{Unger98}); so is
approximating $W^*$, and exact computation is the case $c = 1$.
\end{proof}

\begin{remark}[Fixed window]\label{rem:saxe}
For fixed $W$, seatability is decidable in $n^{\OO(W)}$ time by a
Saxe-style window automaton~\cite{Saxe80}; we sketch it and claim no
more. Scan the slots in order, fixing the rotation
by the packet at slot $0$; a state records the ordered contents of
the last $W$ slots together with each seated packet's unmet
obligations (readers not yet seated, operands not yet seated); a
packet with an obligation older than $W$ slots can never be served,
so such states are dead; there are $n^{\OO(W)}$ states, a transition
seats one packet or leaves a hole, and seatability is reachability of
the state in which every obligation is met and the scan wraps
consistently.
\end{remark}

\subsection{Proof of Theorem~\ref{thm:action}, upper bound}
\label{app:action}
Let $\lambda = \lfloor (W-5)/2 \rfloor \geq 1$ and $b = 2^{\lambda +
1}$, so $b = 4$ at $W = 8$. The post-order linearization of a
balanced tree with $b$ leaves has strand load exactly $\log_2 b - 1 =
\lambda$: at a join, the values produced earlier and still needed are
the pending left results of its ancestors and its own left child's,
less the one that immediately precedes it, and the count peaks at
the last join of the deepest right spine (for $b = 8$, in post-order
$a, b, c{=}a{+}b, d, e, f{=}d{+}e, g{=}c{+}f$: at $f$ the values $c$
and $d$ are earlier and still needed, load $2$). (A natural first construction cuts a
post-order linearization of the whole reduction into segments of $b
\geq \alpha (W + \log N)$ operations. A segment of that length has
strand load about $\log_2 b$, above $\lambda$ for fixed $W$ once $N$
is large, so that construction does not seat. The one below does; it
is the one our compiler implements.)

\emph{The hierarchy.} Cut the $N$-leaf reduction into $N/b$ subtrees
of $b$ leaves and combine their roots by a tree in which every node
again has $b$ children that are roots of the level below: a tree of
$N/b + N/b^2 + \dots = \OO(N/b)$ subtrees, each a balanced tree with
$b$ leaves and $b - 1$ operations. Give each subtree its own ring. By
Theorem~\ref{thm:lattice} its post-order seats on a ring of period
$\OO(b)$ with its $b$ leaves as seeds, and the ring also carries one
transfer instruction at its root's cell. The leaves of a bottom
subtree are seeded; the leaves of a higher subtree are the roots of
its $b$ children, landed by their transfer stations into the seed
slots of the cells that consume them (a settle seating is
rotation-invariant, so each receiving ring is rotated until every
landing falls within the transfer's destination offset, which is the
stage-rotation step of our compiler).

\emph{The run.} Process the subtrees in post-order of the tree of
rings, one at a time; every other ring is parked. An active ring runs
one revolution to settle (Appendix~\ref{app:lattice}) and at most one
more to bring its root to its transfer station: $\OO(b)$ cycles with
$\OO(b)$ live packets, $\OO(b^2)$ rotation action, and one transfer.
Over $\OO(N/b)$ subtrees the rotation action is $\OO(N b)$, the
transfers number $\OO(N/b)$, and the run lasts $\OO(N)$ cycles. The
parked action is at most the number of packets, $\OO(N)$, times the
run length: $\OO(N^2)$. With the lower bound $\act \geq N^2/(4s)$ of
the main text, the rotation-action ratio is $\Omega(N/(s\,b))$,
linear in $N$ for fixed $W$; pricing the parked action at
$p_{\mathrm{park}}$, the hierarchy's bill is $e\,\OO(Nb) +
p_{\mathrm{park}}\,c_2 N^2$ with $c_2$ the construction's parked
constant, which the main text reports on the instrument.

\subsection{Proof of Theorem~\ref{thm:hk} on ring hierarchies}
\label{app:hk}
Take a two-level \csm{} whose compute stations sit on fast-level rings
of total capacity $S$ slots and whose input matrices reside on the
slow level or in the environment; a transfer is a crossing of the
boundary. \emph{Lower bound.} Partition the execution into maximal
segments during which at most $S$ transfers occur. During a segment
the fast level holds at most $2S$ distinct values of each of the three
classes, $A$-entries, $B$-entries and $C$-contributions: at most $S$
resident at its start and at most $S$ imported, with partial results
that are evicted and re-imported accounted exactly as in the classical
proof. Every elementary product is computed at a fast-level station
from fast-resident operands, since a station reads only its own ring,
so by the Loomis--Whitney argument
of~\cite{HongKung81,IronyToledoTiskin04} a segment computes at most
$(2S)^{3/2}$ elementary products. There are therefore at least
$n^3/(2S)^{3/2}$ segments, every full segment performs $S$ transfers,
and the transfers number at least $S\,n^3/(2S)^{3/2} - S =
\Omega(n^3/\sqrt S)$; with $w$-bit words and fare $c_0$ per bit the
energy is $c_0 w\,\Omega(n^3/\sqrt S)$. \emph{Upper bound.} Promote
$\sqrt{S/3} \times \sqrt{S/3}$ tiles of $A$ and $B$ to fast rings,
evaluate the corresponding $C$-tile contributions by streaming inner
products on fixed-period compute rings (our instrument's primitive,
Section~\ref{sec:instrument}), demote, and repeat: $\Theta(S)$ words
move per tile pair and there are $\Theta(n^3/S^{3/2})$ tile pairs,
$\OO(n^3/\sqrt S)$ transfers in all.

\subsection{Proposition~\ref{prop:systolic} (systolic embedding) and its proof}
\label{app:systolic}

\begin{proposition}[Systolic embedding]\label{prop:systolic}
Every unidirectional linear systolic algorithm with $n$ cells whose
inputs and results both stream (a moving-results design in the
classical taxonomy~\cite{Kung82,KungLeiserson78}) embeds on a chain of
$n$ rings, one operation station and one transfer station each, with
window $W = \OO(1)$ on every ring and constant slowdown: a cell is a
station, a beat is a constant-size block carrying its instruction
alongside its operands, a stationary cell constant is a station
constant, and a delay between streams is the phase offset of a
scheduled transfer. A one-cell design embeds on a single ring.
\end{proposition}

\noindent
(A single ring cannot carry $n > 1$ cells at constant window: an
instruction packet reads the same slots at every station it passes,
so the per-cell delays that make an array systolic have nowhere to
live but in the phase between rings.)

A unidirectional linear systolic algorithm with $n$ cells and moving
inputs and results is a schedule of beats: at beat $k$, cell $c$
reads a bounded number $r$ of values of its input streams, at beat
offsets in a fixed set $\{0, \dots, \rho\}$ behind the current beat,
reads at most one cell constant (a stationary weight), and writes one
value of its output stream, which cell $c{+}1$ reads $\delta_c$ beats
later, a constant delay; unidirectional means that the streams flow
from cell $c$ to cell $c{+}1$ only. Give cell $c$ a ring with one
operation station followed, in the ring's direction of travel, by one
transfer station, so that a block is computed before it is forwarded,
and lay its beats along the ring as consecutive blocks of constant
size $\beta$, at most $r$ plus the number of streams the cell
forwards plus three:
the beat's instruction, one transfer instruction per stream the cell
forwards, its $r$ stream operands (landed by the previous cell, or
seeded for cell $1$), its constant if it travels, and its result slot,
the result slot ahead of both instructions that touch it and the
transfer instruction behind the beat's instruction, so that it
reaches the transfer station, downstream of the operation station,
after the result it reads has been written.
The ring advances one slot per cycle, so the blocks ahead of the
instruction are the earlier beats: every read and write of the
instruction lies in its own block or in the $\rho$ blocks ahead, and
$W \leq (\rho + 1)\beta = \OO(1)$. A transfer instruction of block
$k$ fires when the block passes the transfer station and lands the
value it carries in block $k + \delta$ of the next ring, at a landing
position upstream of that ring's operation station so that the value
arrives before the beat that needs it, and at a destination offset
that is the same for every block of its stream, because all rings
advance in lockstep with the same block structure and period: the
delay $\delta_c$ of the result stream is the phase of the next ring
relative to this one, chosen once, and a cell that forwards two
streams at different speeds (inputs and results moving at different
rates, as in Kung's designs) carries two transfer instructions per
block whose two delays are two destination offsets, one per stream.
The operation station fires once per passing block and the transfer
station once per stream the block forwards, so a beat costs $\beta$
cycles on every ring at once: constant slowdown. A cell
constant that must not travel is a station constant, as it is in the
array. For $n = 1$ there is one ring and no transfer, which is our
artifact's streaming tap; its two-stage pipeline kernel is the chain
with $n = 2$.

\end{document}